\documentclass[11p,a4paper]{article}
\usepackage{hyperref}
\usepackage{url}

\usepackage{graphicx} 
\usepackage[round]{natbib}
\usepackage{algorithm}
\usepackage{algorithmic}

\usepackage{amsmath}
\usepackage{amsfonts}
\usepackage{amsthm}
\usepackage{bm}
\usepackage{nicefrac}
\usepackage{enumerate}
\allowdisplaybreaks

\theoremstyle{plain}
\newtheorem{thm}{\protect\theoremname}
\theoremstyle{plain}

\theoremstyle{definition}

\theoremstyle{plain}
\newtheorem{cor}[]{\protect\corollaryname}
\theoremstyle{plain}

\theoremstyle{plain}

\theoremstyle{plain}
\newtheorem{lem}[]{\protect\lemmaname}
\theoremstyle{definition}
\newtheorem{rem}[]{\protect\remarkname}

\providecommand{\algorithmname}{Algorithm}
\providecommand{\assumptionname}{Assumption}
\providecommand{\corollaryname}{Corollary}
\providecommand{\definitionname}{Definition}
\providecommand{\propositionname}{Proposition}
\providecommand{\theoremname}{Theorem}
\providecommand{\remarkname}{Remark}
\providecommand{\lemmaname}{Lemma}

\begin{document}

\title{Optimal Transport: Fast Probabilistic Approximation with Exact Solvers}

\author{Max Sommerfeld\thanks{Supported by the DFG Research Training Group 2088 ``Discovering Structure in Complex Data: Statistics Meets Optimization and Inverse Problems''.}\ \thanks{Felix-Bernstein Institute for Mathematical Statistics in the Biosciences,
University G\"ottingen, Goldschmidtstr. 7, 37077 G\"ottingen}\, , \ 
 J\"orn Schrieber\footnotemark[1] \thanks{Institute for Mathematical Stochastics, University G\"ottingen, Goldschmidtstr. 7, 37077 G\"ottingen}\, ,
 Yoav Zemel\footnotemark[2] \thanks{Supported by Swiss National Science Foundation Grant \#178220}\, ,
 Axel Munk\footnotemark[3] \thanks{Max-Planck-Institute for Biophysical Chemistry,
Am Fa\ss{}berg 11, 37077 G\"ottingen\newline
{Email:} 
{max.sommerfeld@mathematik.uni-goettingen.de}, 
{joern.schrieber-1@mathematik.uni-goettingen.de},
{yoav.zemel@mathematik.uni-goettingen.de},
{munk@math.uni-goettingen.de}\newline
We thank an Associate Editor and three revieweres for insightful comments on a previous version of the paper.}
}

\maketitle

\vspace{-0.02cm}

\begin{abstract}
We propose a simple subsampling scheme for fast randomized approximate computation of optimal transport distances on finite spaces. This scheme operates on a random subset of the full data and can use any exact algorithm as a black-box back-end, including state-of-the-art solvers and entropically penalized versions. It is based on averaging the exact distances between empirical measures generated from independent samples from the original measures and can easily be tuned towards higher accuracy or shorter computation times. 
To this end, we give non-asymptotic deviation bounds for its accuracy in the case of discrete optimal transport problems. In particular, we show that in many important instances, including images (2D-histograms), the approximation error is independent of the size of the full problem.
We present numerical experiments that demonstrate that a very good approximation in typical applications can be obtained in a computation time that is several orders of magnitude smaller than what is required for exact computation of the full problem.
\end{abstract}

\global\long\def\CC{\mathbb{C}}
\global\long\def\SS{S^{1}}
\global\long\def\RR{\mathbb{R}}
\global\long\def\actson{\curvearrowright}
\global\long\def\ra{\rightarrow}
\global\long\def\z{\mathbf{z}}
\global\long\def\ZZ{\mathbb{Z}}
\global\long\def\NN{\mathbb{N}}
\global\long\def\sgn{\mathrm{sgn}\:}
\global\long\def\RRpos{\RR_{>0}}
\global\long\def\var{\mathrm{var}}
\global\long\def\circint{\int_{-\pi}^{\pi}}
\global\long\def\F{\mathcal{F}}
\global\long\def\pb#1{\langle#1\rangle}
\global\long\def\op{\mathrm{op}}
\global\long\def\Op{\mathrm{op}}
\global\long\def\supp{\mathrm{supp}}
\global\long\def\ceil#1{\lceil#1\rceil}
\global\long\def\TV{\mathrm{TV}}
\global\long\def\floor#1{\lfloor#1\rfloor}
\global\long\def\vt{\vartheta}
\global\long\def\vp{\varphi}
\global\long\def\class#1{[#1]}
\global\long\def\of{(\cdot)}
\global\long\def\one{\mathbbm{1}}
\global\long\def\cov{\mathrm{cov}}
\global\long\def\CC{\mathbb{C}}
\global\long\def\SS{S^{1}}
\global\long\def\RR{\mathbb{R}}
\global\long\def\actson{\curvearrowright}
\global\long\def\ra{\rightarrow}
\global\long\def\z{\mathbf{z}}
\global\long\def\ZZ{\mathbb{Z}}
\global\long\def\h{\mu}
\global\long\def\convr{*_{\RR}}
\global\long\def\x{\mathbf{x}}
\global\long\def\ve{\epsilon}
\global\long\def\cv{\mathfrak{c}}
\global\long\def\wh#1{\hat{#1}}
\global\long\def\norm#1{\left|\left|#1\right|\right|}
\global\long\def\Tmean{T}
\global\long\def\Tslope{m}
\global\long\def\degC#1{#1^{\circ}\mathrm{C}}
\global\long\def\X{\mathbf{X}}
\global\long\def\bbeta{\boldsymbol{\beta}}
\global\long\def\b{\mathbf{b}}
\global\long\def\Y{\mathbf{Y}}
\global\long\def\H{\mathbf{H}}
\global\long\def\e{\bm{\epsilon}}
\global\long\def\s{\mathbf{s}}
\global\long\def\t{\mathbf{t}}
\global\long\def\R{\mathbf{R}}
\global\long\def\dAc{\partial A_c}
\global\long\def\cl{\mathrm{cl}}
\global\long\def\Hoe{\text{H\"older}}
\global\long\def\bb#1{\mathbb{#1}}
\global\long\def\bX{\bm{X}}
\global\long\def\bY{\bm{Y}}
\global\long\def\D{\mathfrak{D}}
\global\long\def\dGH{d_{\mathcal{GH}}}
\global\long\def\dWp{d_{\mathcal{W},p}}
\global\long\def\umu{\underline{\mu}}
\global\long\def\FLB{\mathbf{FLB}}
\global\long\def\dX{\mathbf{d_\bX}}
\global\long\def\dY{\mathbf{d_\bY}}
\global\long\def\mX{{\mu_\bX}}
\global\long\def\mY{{\mu_\bY}}	
\global\long\def\cal#1{\mathbf{\mathcal{#1}}}
\global\long\def\UXn{U_{\bX n}}
\global\long\def\dXX{{\Delta_{\bX}}}
\global\long\def\dYY{{\Delta_{\bY}}}
\global\long\def\d{\:\mathrm{d}}
\global\long\def\Ds{\mathfrak{d}}
\global\long\def\ph{\hat{p}}
\global\long\def\muh{\hat{\mu}}
\global\long\def\onevec{\one}
\global\long\def\T{\mathcal{T}}
\global\long\def\root{\mathrm{root}}
\global\long\def\parent{\mathrm{par}}
\global\long\def\children{\mathrm{children}}
\global\long\def\scp#1#2{\langle #1, #2 \rangle}
\global\long\def\X{\mathcal{X}}
\global\long\def\br{\bm{r}}
\global\long\def\bs{\bm{s}}
\global\long\def\brh{\hat{\bm{r}}}
\global\long\def\bsh{\hat{\bm{s}}}
\global\long\def\rh{\hat{r}}
\global\long\def\sh{\hat{s}}
\global\long\def\bu{{\bm{u}}}
\global\long\def\bv{\bm{v}}
\global\long\def\bh{\bm{h}}
\global\long\def\bw{\bm{w}}
\global\long\def\bD{{\bm{D}}}
\global\long\def\bx{\bm{x}}
\global\long\def\by{\bm{y}}
\global\long\def\bZ{\bm{Z}}
\newcommand{\WhnmB}{\hat{W}^{(S)}}
\global\long\def\smax{{\sigma_{\mathrm{max}}}}
\global\long\def\sT{\sigma_\T}
\global\long\def\bsT{\bm{\sigma}_\T}
\global\long\def\E{\mathcal{E}}
\global\long\def\V{\mathcal{V}}
\global\long\def\cmax{c_{\mathrm{max}}}
\global\long\def\h{h}
\global\long\def\Id{\mathrm{Id}}
\global\long\def\N{\mathcal{N}}
\global\long\def\diam{\:\mathrm{diam}}
\global\long\def\ns{{n^{*}}}
\global\long\def\ps{{p^{*}}}
\global\long\def\Wt{W^{(t)}}
\global\long\def\BL1{\mathrm{BL}_1}
\global\long\def\I{\mathcal{I}}
\global\long\def\simD{\stackrel{D}{\sim}}
\global\long\def\d{\mathrm{d}}
\global\long\def\dpp{d^p\!}

\section{Introduction}
\label{sec:intro}

 Optimal transport distances, a.k.a.\ Wasserstein, earth-mover's,
 Monge-Kantorovich-Rubinstein or Mallows distances,   
 as  metrics to compare probability measures \citep{rachev_rueschendorf_1998, villani_optimal_2008}
 have become a popular tool in a wide range of applications in computer science,
 machine learning and statistics. 
 Important examples are image retrieval  \citep{rubner_earth_2000} and classification
 \citep{zhang_local_2007},  computer vision \citep{ni_local_2009},
 but also therapeutic equivalence \citep{Munk1998}, generative modeling \citep{bousquet_2017},
 biometrics \citep{sommerfeld_inference_2018}, metagenomics
 \citep{evans_phylogenetic_2012} and medical imaging \citep{Ruttenberg2013}.

Optimal transport distances compare probability measures by incorporating a suitable
ground distance on the underlying space, typically driven by the particular application, e.g. euclidean distance. This often makes it preferable to
competing distances such as total-variation or $\chi^2$-distances, which are
oblivious to any metric or similarity structure on the ground space.  Note that total variation is the Wasserstein distance with respect to the trivial metric, which usually does not carry the geometry of the underlying ground space.  In this setting, optimal transport distances have a clear and intuitive interpretation
as the amount of `work' required to transport one probability distribution onto
the other. 
This notion is typically well-aligned with human perception of
similarity \citep{rubner_earth_2000}.

\begin{figure}
  \centering
  \includegraphics[width=0.45\textwidth]{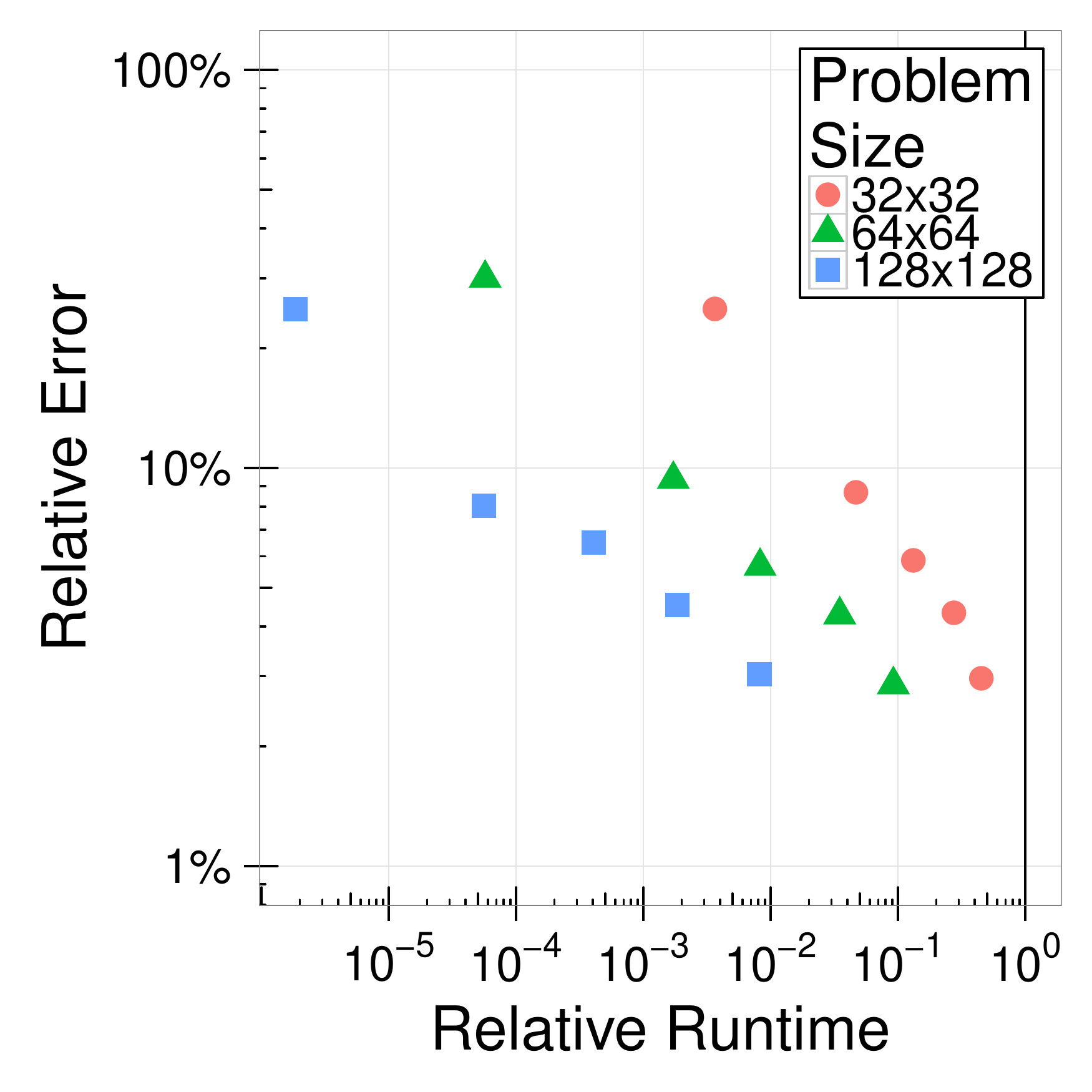}
  \caption{Relative error and relative runtime compared to the exact computation
  of the proposed scheme. Optimal transport distances and its approximations
were computed between images of different sizes ($32\times32$, $64\times 64$,
$128\times 128$). Each point represents a specific parameter choice in the
scheme and is a mean over different problem instances, solvers and cost
exponents. For the relative runtimes the geometric mean is reported. For details
on the parameters see Figure \ref{fig:overview}. }
  \label{fig:meanrelrel}
\end{figure}

\subsection{Computation}
The outstanding theoretical and practical performance of optimal transport
distances is contrasted by its excessive computational cost. 
For example, optimal transport distances can be computed with an auction
algorithm \citep{bertsekas_auction_1992}. 
For  two probability
measures supported on $N$ points this algorithm has a worst case run time of
$\mathcal{O}(N^3\log N)$. Other methods like the transportation simplex have
sub-cubic empirical average runtime (compare \cite{Gottschlich2014}), but 
exponential worst case runtimes.

Therefore, many attempts have been made to design improved algorithms. We give some selective references: \cite{ling_efficient_2007} proposed a specialized algorithm for $L_1$-ground distance and $\X$ a regular grid and report an empirical runtime of $\mathcal{O}(N^2)$.\cite{Gottschlich2014} improved existing general purpose algorithms by initializing with a greedy heuristic. Their \textit{Shortlist} algorithm achieves an empirical average runtime of the order
$\mathcal{O}(N^{5/2})$. \cite{schmitzer_sparse_2016} solves the optimal transport problem by solving a sequence of sparse problems.  The theoretical runtime of his algorithm is not known, but it exhibits excellent performance on two-dimensional grids \citep{schrieber_dotmark_2016}. The literature on this topic is rapidly growing and we refer for further recent work to \citet{liu2018multilevel}, \citet{dvurechensky2018computational}, \citet{lin2019efficient}, and the references given there.

Despite these efforts, still many practically relevant problems remain well outside the scope of available algorithms. See \cite{schrieber_dotmark_2016} for an overview and a numerical comparison of state-of-the-art algorithms for discrete optimal transport. This is true in particular for two or three dimensional images and spatio temporal imaging, which constitute an important area of potential applications. Here,  $N$ is  the number of pixels or voxels and is  typically of size $10^5$ to $10^7$. Naturally, this problem is aggravated when many distances have to be computed as is the case for Wasserstein barycenters \citep{agueh_barycenters_2011, cuturi_fast_2014}, which have become an important use case.

To bypass the computational bottleneck, also many surrogates for optimal transport
distances that are more amenable to fast computation have been proposed. 
\cite{shirdhonkar_approximate_2008} proposed to use an equivalent distance based on wavelets
that can be computed in linear time but cannot be calibrated to approximate the
Wasserstein distance with arbitrary accuracy.
\cite{pele_fast_2009} threshold the ground distance to reduce the complexity of
the underlying linear program, obtaining a lower bound for the exact distance.
\cite{Cuturi2013a} altered the optimization problem by adding an entropic
penalty term in order to use faster and more stable algorithms, see also \cite{altschuler_greenkhorn_2017}.
\cite{bonneel_sliced_2015} consider the 1-D Wasserstein distances of radial projections
of the original measures, exploiting the fact that, in one dimension, computing the
Wasserstein distance  amounts to sorting the point masses and hence has
quasi-linear computation time.

\subsection{Contribution}
We do \textit{not} propose a new algorithm to solve the optimal transport
problem. Instead, we propose a simple probabilistic scheme as a meta-algorithm that can use any
algorithm (e.g., those mentioned above) solving finitely supported optimal transport problems  as a black-box back-end and gives a random but fast approximation of the exact distance. This scheme
\begin{enumerate}[a)]
  \item is extremely easy to implement, to parallelize and to tune towards higher accuracy or
    shorter computation time as desired;
  \item can be used with any algorithm for transportation problems as a
    back-end, including general LP solvers, specialized network solvers and
    algorithms using entropic penalization \citep{Cuturi2013a};
  \item comes with theoretical non-asymptotic guarantees for the approximation
    error of the Wasserstein distance  --- in particular, this error is independent of the size of the original
    problem in many important cases, including images;
  \item works well in practice. For example, the Wasserstein distance between
    two $128^2$-pixel images can typically be approximated with a relative
    error of less than $5\%$ in only $1\%$ of the time required
    for exact computation.
\end{enumerate}

\section{Problem and Algorithm}
\label{sec:theory}
Although our meta-algorithm is applicable to exact solvers for any optimal
transport distance between probability measures, for example the Sinkhorn distance \citep{Cuturi2013a},
the theory we present here concerns the \cite{kantorovich_translocation_1942} transport distance, often also denoted as \textit{Wasserstein distance}.

\paragraph{Wasserstein Distance}
Consider a fixed finite space $\X = \left\{ x_1,\dots,x_N \right\}$ with a
metric $d:\X\times \X \ra [0, \infty)$. Every
probability measure on $\X$ is given by a vector $\br$ in  
\[
  \mathcal{P}_\X = \left\{ \br = (r_x)_{x\in\X} \in\RR_{\geq 0}^\X : 
  \sum_{x\in\X} r_x =1 \right\},
\]
via $P_{\br}(\{x\}) = r_x$. We will not distinguish between the vector
$\br$ and the measure it defines.
For $p\geq 1$, the \textit{$p$-th Wasserstein distance} between two probability measures
$\br,\bs \in\mathcal{P}_\X$  is defined as
\begin{equation} \label{eq:def_wass}
  W_p(\br, \bs) = 
  \left(  
    \min_{\bw\in\Pi(\br, \bs)}  \sum_{x,x'\in\X} \dpp(x, x') w_{x,x'}
\right)^{1/p},
\end{equation}
where $\Pi(\br, \bs)$ is the set of all probability measures  on $\X\times\X$
with marginal distributions $\br$ and $\bs$, respectively. 
The minimization in \eqref{eq:def_wass} can be written as a linear program
\begin{equation}
  \label{eq:lp}
   \min \sum_{x,x'\in\X} w_{x,x'} \dpp(x,x') 
  \quad \textbf{s.t.} \quad 
  \sum_{x'\in\X} w_{x,x'} = r_x, \quad \sum_{x\in\X} w_{x,x'} = s_{x'},
  \quad w_{x,x'} \geq 0, 
\end{equation}
with $N^2$ variables $w_{x,x'}$ and $2 N$ constraints, where
the weights $\dpp(x, x')$ are known and have been precalculated.

\subsection{Approximating the Wasserstein Distance}
The idea of the proposed algorithm is to replace a probability measure
$\br\in\mathcal{P}(\X)$ with an empirical measure $\brh_S$
based on i.i.d.\ picks $X_1, \dots,
X_S\sim\br$ for some integer $S$:
\begin{equation}
  \rh_{S,x} = \frac{1}{S}\#\left\{ k : X_k=x \right\}, \quad x \in \X.
  \label{eq:rhat}
\end{equation}
Likewise, replace $\bs$ with $\bsh_S$. Then, use the \emph{empirical optimal transport distance} (EOT) $W_p(\brh_S, \bsh_S)$ as a random approximation of $W_p(\br, \bs)$.

\begin{algorithm}[H]
  \caption{Statistical approximation of $W_p(\br, \bs)$}
  \label{Subsamplealgo}
  \begin{algorithmic}[1]
    \STATE {\bfseries Input:} Probability measures $\br, \bs \in\mathcal{P}_\X$, sample
  size $S$ and number of repetitions $B$
  \FOR{$i = 1 \dots B$}
    \STATE Sample i.i.d. $X_1 ,\dots, X_S\sim\br$ and independently $Y_1, \dots, Y_S\sim
    \bs$
    \STATE $\rh_{S,x} \gets \#\left\{ k : X_k=x \right\}/S$ \textbf{for all}
    $x\in\X$
    \STATE $\sh_{S,x} \gets \#\left\{ k : Y_k=x \right\}/S$ \textbf{for all}
    $x\in\X$
    \STATE Compute $\hat{W}^{(i)} \gets W_p(\brh_S, \bsh_S)$
    \ENDFOR
    \STATE {\bfseries Return:} $\WhnmB_p(\br, \bs) \gets B^{-1} \sum_{i = 1}^B \hat{W}^{(i)}$
  \end{algorithmic}
\end{algorithm}
In each of the $B$ iterations in Algorithm \ref{Subsamplealgo}, the Wasserstein
distance between two sets of $S$ point masses has to be computed. For the exact
Wasserstein distance, two measures on $N$ points need to be compared. If we take
for example the super-cubic runtime of the auction algorithm as a basis, Algorithm
\ref{Subsamplealgo} has worst case runtime
\[
  \mathcal{O}(BS^3 \log S)
\]
compared to $\mathcal{O}(N^3\log N)$ for the exact distance. This means a dramatic reduction
of computation time if $S$ (and $B$) are small compared to $N$.

The application of Algorithm \ref{Subsamplealgo} to other optimal transport
distances is straightforward. One can simply replace $W_p(\brh_S, \bsh_S)$ with
the desired distance, e.g., the Sinkhorn distance \citep{Cuturi2013a}, see also
our numerical experiments below. Further, the algorithm can be applied to non-discrete instances as long as we can sample from the measures. However, the theoretical results below only apply to the EOT on a finite ground space $\X$.

\section{Theoretical results}
\label{sec:theory_app}
We give general non-asymptotic guarantees for the quality of the
approximation $\WhnmB_p(\br, \bs)=B^{-1}\sum_{i=1}^B W_p(\brh_{S,i},\bsh_{S,i})$ (where $\brh_{S,i}$ are independent empirical measures of size $S$ from $\br$;  see Algorithm~\ref{Subsamplealgo}) in terms of the expected $L_1$-error. That is, we give bounds of the form
\begin{equation}
  E\left[ \left| \WhnmB_p(\br, \bs) - W_p(\br, \bs) \right| \right] \leq g(S,
  \X, p),
  \label{eq:expected_L1_error}
\end{equation}
for some function $g$. We are particularly interested in the dependence of the
bound on the size $N$ of $\X$ and on the sample size $S$ as this
determines how the number of sampling points $S$ (and hence the computational effort
of Algorithm \ref{Subsamplealgo}) must be increased for increasing problem size
$N$ in order to retain (on average) a certain approximation quality.
In a second step, we obtain deviation inequalities for $\WhnmB(\br,
\bs)$ via concentration of measure techniques.

\paragraph{Related work} 
The question of the convergence of empirical measures to the true measure in
expected Wasserstein distance has been considered in detail by
\cite{Boissard2014} and \cite{fournier_rate_2014}. The case of the underlying
measures being different (that is, the convergence of $EW_p(\brh_S, \bsh_S)$ to
$W_p(\br, \bs)$ when $\br\neq\bs$) has not been considered to the best of our knowledge.
Theorem \ref{thm:mean_rate_null} is reminiscent of the main result
of \cite{Boissard2014}. However, we give a result here, which is explicitly
tailored to finite spaces and makes explicit the dependence of the constants on
the size $N$ of the underlying set $\X$.
In fact, when we consider finite spaces $\X$ which are subsets of $\RR^D$ later
in Theorem \ref{thm:grid}, we will see that in contrast to the results of
\cite{Boissard2014}, the rate of convergence (in $S$) does not change when the
dimension gets large, but rather the dependence of the constants on $N$ changes.
This is a valuable insight as our main concern here is how the subsample size
$S$ (driving the computational cost) must be chosen when $N$ grows in order to
retain a certain approximation quality.

\subsection{Expected absolute error}
Recall that, for $\delta>0$ the \textit{covering number} $\mathcal{N}(\X,\delta)$ of $\X$ is defined as the minimal number of closed balls with radius $\delta$ and centers in $\X$ that is needed to cover $\X$. Note that in contrast to continuous spaces, $\mathcal N(\X,\delta)$ is bounded by$N$ for all $\delta > 0$.

\begin{thm}
  \label{thm:mean_rate_null}
  Let $\brh_S$ be the empirical measure obtained from i.i.d. samples $X_1, \dots,
  X_S\sim \br$, then
  \begin{equation}
    E\left[ W_p^p(\brh_S, \br)\right] 
    \leq   \E_q / \sqrt{S}, 
    \label{eq:mean_rate_null}
  \end{equation}
  where the constant $\E_q := \E_q(\X, p)$ is given
  by
  \begin{equation}
    \begin{split}
    \E_q  = 
    2^{p-1}q^{2p} (\diam(\X))^p
    \left( q^{-(l_{\max} + 1)p} \sqrt{N} + \sum_{l=0}^{l_{\max}} q^{-lp}
    \sqrt{\mathcal{N}(\X, q^{-l}\diam(\X))} \right)
  \end{split}
     \label{eq:Eq}
  \end{equation}
  for any $2 \leq q\in \NN$ and $l_{\max}\in \NN$.
  \end{thm}

\begin{rem}
Since Theorem \ref{thm:mean_rate_null} holds for any integer $q \ge 2$ and $l_{\max} \in \NN$, they can be chosen freely to minimize the constant $\E_q$. In the proof they appear as the branching number and depth of a spanning tree that is constructed on $\X$ (see appendix).  In general, an optimal choice of $q$ and $l_{\max}$ cannot be given.  However, in the Euclidean case, the optimal values for $q$ and $l_{\max}$ will be determined, and in particular we will show that $q=2$ is optimal (see the discussion after Theorem~\ref{thm:grid}, and Lemma~\ref{lem:q2}).
\end{rem}

\begin{rem}[covering by arbitrary sets]\label{rem:N1}
At the price of a factor $2^p$, we can replace the balls defining the covering numbers $\mathcal N$ with arbitrary sets, and obtain the bound
\[
    \E_q  = 
    2^{2p-1}q^{2p} (\diam(\X))^p
    \left( q^{-(l_{\max} + 1)p} \sqrt{N} + \sum_{l=0}^{l_{\max}} q^{-lp}
    \sqrt{\mathcal{N}_1(\X, q^{-l}\diam(\X))} \right),
\]
where $\mathcal N_1(\X,\delta)$ is the minimal number of closed sets of diameter $\le2\delta$ needed to cover $\X$.  The proof is given in the appendix.  These alternative covering numbers lead to better bounds in high-dimensional Euclidean spaces when $p>2.5$ (see Remark~\ref{rem:verger}).
\end{rem}

Based on Theorem \ref{thm:mean_rate_null}, we can formulate a bound for the mean
approximation error of Algorithm \ref{Subsamplealgo}.  A mean squared error version is given below, in Theorem~\ref{thm:mean_squared_alt}.

\begin{thm}
  \label{thm:mean_rate_alt}
  Let $\WhnmB_p(\br, \bs)$ be as in Algorithm \ref{Subsamplealgo} for any choice of $B \in \NN$.
  Then for every integer $q\geq 2$
  \begin{equation}
    E\left[ \left| \WhnmB_p(\br, \bs) - W_p(\br, \bs) \right| \right] 
    \leq 
    2\E_q^{1/p} S^{-1/(2p)}.
    \label{eq:mean_app_err}
  \end{equation}
\end{thm}

  \begin{proof}
  The statement is an immediate consequence of the reverse triangle inequality for the Wasserstein distance, Jensen's inequality and Theorem \ref{thm:mean_rate_null},
  \begin{multline*}
    E\left[ \left| \WhnmB_p(\br, \bs) - W_p(\br, \bs)\right| \right] \leq 
    E\left[ W_p(\brh_S, \br) + W_p(\bsh_S, \bs) \right]
    \\  \leq E\left[ W_p^p(\brh_S, \br) \right]^{1/p} + E\left[W_p^p(\bsh_S, \bs) \right]^{1/p}
    \leq 2\E_q^{1/p} / S^{1/(2p)}.
  \end{multline*}
\end{proof}

\paragraph{Measures on Euclidean Space}
While the constant $\E_q$ in Theorem \ref{thm:mean_rate_null} may be difficult to compute or estimate in general, we give explicit bounds in the case when $\X$ is a finite subset of a Euclidean space. They exhibit the dependence of the approximation error on $N=|\X|$. In particular, it comprises the case when the measures represent images (two- or more dimensional).

\begin{thm}
  \label{thm:grid}
  Let $\X$ be a finite subset of $\RR^D$ with the usual Euclidean metric. Then,
  \begin{equation*}
    \E_2 \le 
D^{p/2} 2^{3p-1}(\diam(\X))^p \cdot C_{D,p}(N),
  \end{equation*}
      where $N=|\X|$ and
  \begin{equation}\label{eq:CDpN}
    C_{D,p}(N) = 
    \begin{cases}
      1 / (1 - 2^{D/2 - p} ) & \text{if } D < 2p, \\
      2 + D^{-1}\log_2N & \text{if } D = 2p, \\
      N^{1/2 - p/D}[2+1/(2^{D/2 - p} - 1)] & \text{if } D > 2p.
    \end{cases}
  \end{equation}
\end{thm}
One can obtain bounds for $\E_q$, $q>2$ (see the proof), but the choice $q=2$ leads to the smallest bound (Lemma~\ref{lem:q2}(a), page \pageref{lem:q2}).  Further, if $p$ is an integer, then
\[
    C_{D,p}(N) \le 
    \begin{cases}
      2 + \sqrt 2 & \text{if } D < 2p, \\
      2 + D^{-1}\log_2N & \text{if } D = 2p, \\
      (3+\sqrt 2)N^{1/2 - p/D} & \text{if } D > 2p
    \end{cases}
\]
(see Lemma~\ref{lem:q2}(b).)

In particular, we have for the most important cases $p=1,2$:
\begin{cor}
Under the conditions of Theorem~\ref{thm:grid},
  \begin{align*}
    p=1 \quad \Longrightarrow \quad &\E_2 \le 
    4D^{1/2}\diam(\X) \cdot
    \begin{cases}
      1 / (1 - 2^{D/2 - 1} )  & \text{if } D < 2, \\
      2 + (1/2)\log_2N  & \text{if } D=2, \\
      N^{1/2 - 1/D}[2+1/(2^{D/2 - 1} - 1)]   & \text{if } D>2.
    \end{cases}\\
    p=2 \quad \Longrightarrow \quad &\E_2 \le 
    32D(\diam(\X))^2 \cdot
    \begin{cases}
      1 / (1 - 2^{D/2 - 2} )  & \text{if } D < 4, \\
      2 + (1/4)\log_2N  & \text{if } D=4, \\
      N^{1/2 - 2/D}[2+1/(2^{D/2 - 2} - 1)]   & \text{if } D>4.
    \end{cases}
  \end{align*}
\end{cor}
\begin{rem}[improved bounds in high dimensions]\label{rem:verger}
The term $D^{p/2}$ appears because in the proof of Theorem~\ref{thm:grid} we switch between the Euclidean norm and the supremum norm.  One may wonder whether this change of norms is necessary.  We can stay in the Euclidean setting, and may assume without loss of generality that $\X$ is included in $B_{\diam(\X)}(0)$, where $B_r(x)=\{y:\|y-x\|_2\le r\}$ is the closed ball of radius $r$ around $x$.  According to \cite{verger2005covering}, there exists an absolute constant $C$ such that $\mathcal N(B_1(0),\epsilon)\le C^2D^{5/2}\epsilon^{-D}$.  Using this would allow to replace $D^{p/2}$ by $C2^{D/2}D^{5/4}$, or, combining the alternative covering numbers $\mathcal N_1$ (Remark~\ref{rem:N1}), by $C2^pD^{5/4}$.  This is better than $D^{p/2}$ when $p>2.5$ and $D$ is large.
\end{rem}

Theorem~\ref{thm:grid} gives control over the error made by the approximation $\WhnmB_p(\br,\bs)$ of $W_p(\br, \bs)$. Of particular interest is the behavior of this error as $N$ gets large (e.g., for high resolution images).  We distinguish three cases. In the \textit{low-dimensional case} $p'=D/2 - p<0$, we have $C_{D,p}(N) = \mathcal{O}(1)$ and the approximation error is $\mathcal{O}(S^{-\frac{1}{2p}})$ independent of the size of the image. In the \textit{critical case} $p' = 0$ the approximation error is no longer independent of $N$ but is of order $\mathcal{O}\left(
\log(N)S^{-\frac{1}{2p}} \right)$. Finally, in the \textit{high-dimensional case} the dependence on $N$ becomes stronger with an approximation error of order 
\[
\mathcal{O}
\left( 
\left(\frac{N^{(1- \frac{2p}{D})}}{S}\right)^{\frac{1}{2p}}
\right).
\]
In all cases one can choose $S=o(N)$ while still guaranteeing vanishing approximation error for $N\ra \infty$. In practice, this means that $S$ can typically be chosen (much) smaller than $N$ to obtain a good approximation of the Wasserstein distance.  In particular, this implies that for low-dimensional applications with two or three dimensional histograms (for example greyscale images, where $N$ corresponds to the number of pixels / voxels  and $\br, \bs$ correspond to the grey value distribution after normalization), the approximation error is essentially not affected by the size of the problem when $p$ is not too small, e.g., $p=2$.

While the three cases in Theorem \ref{thm:grid} resemble those given by \cite{Boissard2014},
the rate of convergence in $S$ as seen in Theorem \ref{thm:mean_rate_null} is
$\mathcal{O}(S^{-1/2})$, regardless of the dimension of the underlying space $\X$.  The constant depends on $D$, however, roughly at the polynomial rate $D^{p/2}$ and through $C_{D,p}(N)$.  It is also worth mentioning that by considering the dual transport problem, one can invoke the framework of \cite{shalev_learnability_2010}, particularly Theorem~7.  However, the dependence on $S$ and $N$ and the constants are not easily accessible from that paper.
\begin{rem}
The results presented here extend to the case where $\X$ is a bounded, countable subset of $\RR^D$.  However, our bounds for $\E_q$ contain the term $C_{D,p}(N)$, which is finite as $N\to\infty$ in the low-dimensional case ($D< 2p$) but infinite otherwise.  Finding a better bound for $\E_q$ when $\X$ is countable is challenging and an interesting topic for further research.
\end{rem}

\subsection{Concentration bounds}
\label{sec:conc}
Based on the bounds for the expected approximation error we now give
non-asymptotic guarantees for the approximation error in the form of deviation
bounds using standard concentration of measure techniques.
\begin{thm}
    \label{thm:simple_concentration}
    If $\WhnmB_p(\br,\bs)$ is obtained from Algorithm \ref{Subsamplealgo}, then
    for every $z\geq 0$
    \begin{equation}
         P\left[ |\WhnmB_p(\br, \bs) - W_p(\br,\bs)| \geq z +
        \frac{2\E_q^{1/p}}{S^{\nicefrac{1}{2p}}} \right] 
         \leq 
        2 \exp\left(
        -\frac{SBz^{2p}}{8\:\diam(\X)^{2p}} \right).
      \label{eq:simple_concentration}
    \end{equation}
\end{thm}
Note that while the mean approximation quality $2\E_q^{1/p} / S^{1/(2p)}$ only
depends on the subsample size $S$, the stochastic variability (see the right
hand side term in \eqref{eq:simple_concentration}) depends on the product
$SB$. This means that the repetition number $B$ cannot decrease the expected
error but it decreases the magnitude of fluctuation around it.

From these concentration bounds we can obtain a mean squared error version of Theorem~\ref{thm:mean_rate_alt}:
\begin{thm}
  \label{thm:mean_squared_alt}
  Let $\WhnmB_p(\br, \bs)$ be as in Algorithm \ref{Subsamplealgo} for any choice of $B \in \NN$. Then for every integer $q\ge 2$ the mean squared error of the EOT can be bounded as 
\[
    E\left[ \left| \WhnmB_p(\br, \bs) - W_p(\br, \bs) \right| ^2\right] 
    \le 18\E_q^{2/p}S^{-1/p}
        =\mathcal O(S^{-1/p}).
\]
\end{thm}
\begin{rem}
The power 2 can be replaced by any $\alpha\le 2p$ with rate $S^{-\alpha/(2p)}$, as can be seen from a straightforward modification of the first lines of the proof.
\end{rem}
For example, in view of Theorem~\ref{thm:grid}, when $\X$ is a finite subset of a $\RR^D$ and $q=2$, we obtain
\[
    E\left[ \left| \WhnmB_p(\br, \bs) - W_p(\br, \bs) \right| ^2\right] 
    \le 3^22^{7-2/p} D C_{D,p}^{2/p}(N)[\diam(\X)]^2S^{-1/p}.
\]
with the constant $C_{D,p}(N)$ given in \eqref{eq:CDpN}.  Thus, we qualitatively observe the same dependence on $N$ as in Theorem~\ref{thm:grid}, e.g., the mean squared error is independent of $N$ when $D<2p$.

\section{Simulations}

This section covers the numerical findings of the simulations. Runtimes and returned values of
Algorithm \ref{Subsamplealgo} for each back-end solver are reported in relation to the results of that solver on the original problem. Four different solvers are tested.

\subsection{Simulation Setup}

The setup of our simulations is identical to that of \cite{schrieber_dotmark_2016}. One single core of a Linux server (AMD Opteron Processor 6140 from 2011 with 2.6 GHz) was used. The original and subsampled instances were run under the same conditions.

Three of the four methods featured in this simulation are exact linear programming solvers. The transportation simplex is a modified version of the network simplex solver tailored towards optimal transport problems. Details can be found for example in \cite{Luenberger2008}. The shortlist method \citep{Gottschlich2014} is a modification of the transportation simplex, that performs an additional greedy step to quickly find a good initial solution. The parameters were chosen as the default parameters described in that paper. The third method is the network simplex
solver of CPLEX (\url{www.ibm.com/software/commerce/optimization/cplex-optimizer/}). For the transportation simplex and the shortlist method the implementations provided in the R package \textit{transport} \citep{Schuhmacher2014} were used. The models for the CPLEX solver were created and solved via the R package \textit{Rcplex} \citep{Rcplex}.

Additionally, the Sinkhorn scaling algorithm \citep{Cuturi2013a} was tested in our simulation. This method computes an entropy regularized optimal transport distance. The regularization parameter was chosen according to the heuristic in \cite{Cuturi2013a}. Note that the Sinkhorn distance is not covered by the theoretical results from Section \ref{sec:theory_app}. The errors reported for the Sinkhorn scaling are relative to the values returned by the algorithm on the full problems, which themselves differ from the actual Wasserstein distances.

The instances of optimal transport considered here are discrete instances of two different types: regular grids in two dimensions, that means images in various resolutions, as well as point clouds in $\left[0,1\right]^D$ with dimensions $D = 2$, $3$ and $4$. For the image case, from the DOTmark, which contains images of various types intended to be used as optimal transport instances in the form of two-dimensional histograms, three instances were chosen: two images of each of the classes White Noise, Cauchy Density, and Classic Images, which are then treated in the three resolutions $32 \times 32$, $64 \times 64$ and $128 \times 128$.  Images are interpreted as finitely supported measures.  The mass of a pixel is given by the greyscale value and the support of the measure is the grid $\{1,\dots, R\} \times \{ 1, \dots, R\}$ for an image with resolution $R \times R$.

In the White Noise class the grayscale values of the pixels are independent of each
other, the Cauchy Density images show bivariate Cauchy densities with random
centers and varying scale ellipses, while Classic Images contains grayscale test
images. See \cite{schrieber_dotmark_2016} for further details on the different
image classes and example images. The instances were chosen to cover different
types of images, while still allowing for the simulation of a large variety of
parameters for subsampling.

The point cloud type instances were created as follows: The support points of the measures
are independently, uniformly distributed on $\left[0,1\right]^D$. The number of points $N$
was chosen $32^2$, $64^2$ and $128^2$ in order to match the size of the grid based instances.
For each choice of $D$ and $N$, three instances were generated with regards to the three
images types used in the grid based case. Two measures on the points are drawn from the
Dirichlet distribution with all parameters equal to one. That means, the masses on different points are independent of each other, similar to the white noise images. To create point cloud versions of the Cauchy Density and Classic Images classes the grayscale values of the same images were used to get the mass values for the support points. In three and four dimensions, the product measure of the images with their sum of columns and with themselves, respectively, was used.

All original instances were solved by each back-end solver in each resolution for the values $p = 1$, $p = 2$, and $p = 3$ in order to be compared to the approximative results for the subsamples in terms of runtime and accuracy, with the exception of CPLEX, where the $128 \times 128$ instances could not be solved due to memory limitations. Algorithm \ref{Subsamplealgo} was applied to each of these instances with parameters $S \in \{100, 500, 1000, 2000, 4000\}$ and $B \in \{1, 2, 5\}$. For every combination of instance and parameters, the subsampling algorithm was run $5$ times in order to mitigate the randomness of the results.

Since the linear programming solvers had a very similar performance on the grid based instances (see below), only one of them - the transportation simplex - was tested on the point cloud instances.

\subsection{Computational Results}

\begin{figure*}
  \centering
  \includegraphics[width=0.95\textwidth]{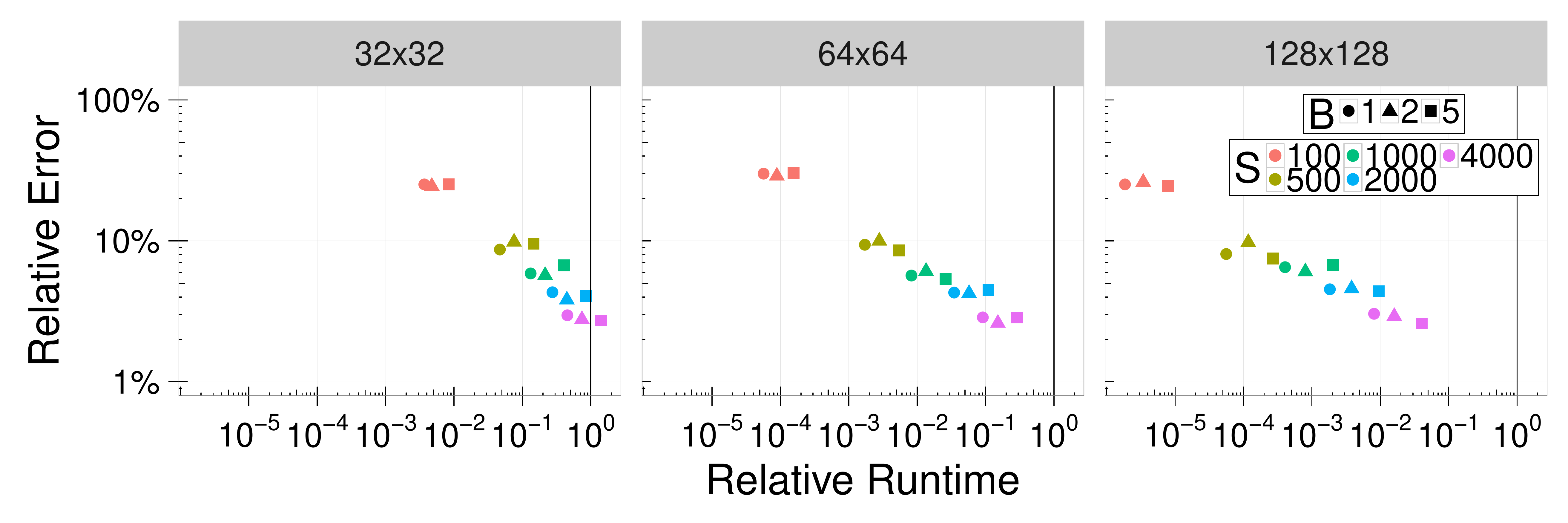}
  \caption {Relative errors $\lvert \WhnmB_p(\br, \bs) -
    W_p(\br, \bs) \rvert / W_p(\br, \bs)$ vs.\ relative runtimes $\hat t / t$ for different parameters
  $S$ and $B$ and different problem sizes for images. $\hat t$ is the runtime of Algorithm \ref{Subsamplealgo}
    and $t$ is the runtime of the respective back-end solver without subsampling.}
  \label{fig:overview}
\end{figure*}

\begin{figure*}
  \centering
  \includegraphics[width=0.95\textwidth]{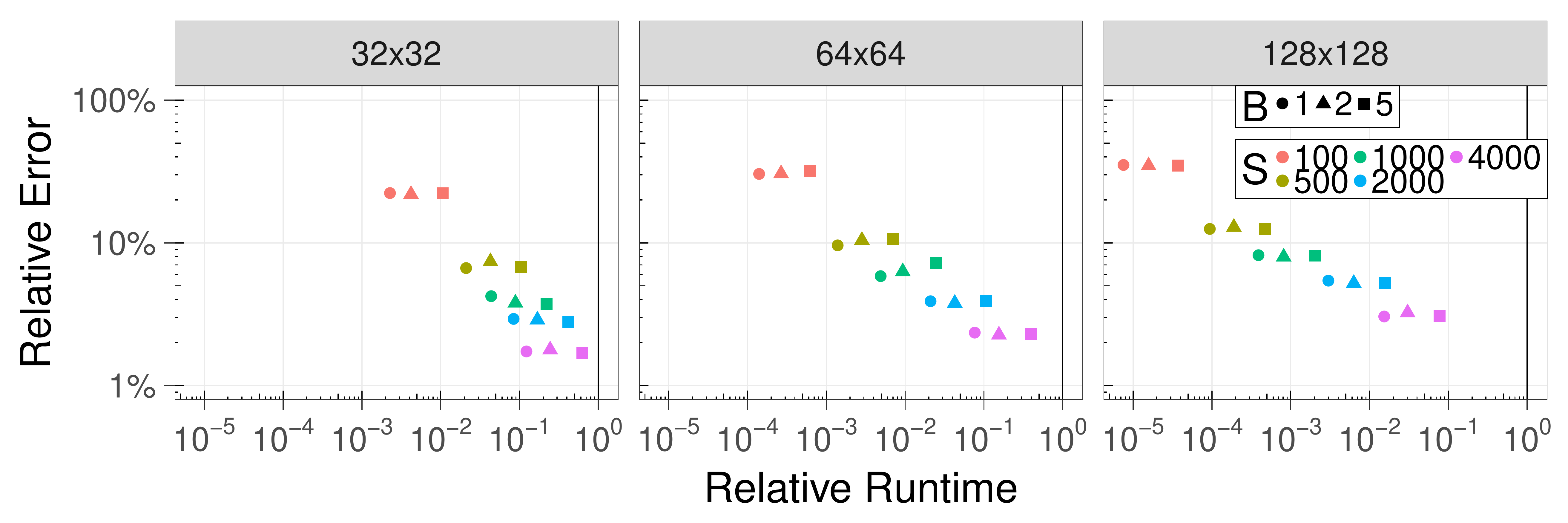}
  \caption{Relative errors vs.\ relative runtimes for different parameters
  $S$ and $B$ and different problem sizes for point clouds.
 The number of support points matches the number of pixels in the images.}
  \label{fig:PC_overview}
\end{figure*}

As mentioned before, all results of Algorithm \ref{Subsamplealgo} are relative
to the results of the methods applied to the original problems. We are mainly
interested in the reduction in runtime and accuracy of the returned values. Many
important results can be observed in Figure \ref{fig:overview} and \ref{fig:PC_overview}. The points in
the diagram represent averages over the different methods, instances, and multiple
tries, but are separated in resolution and choices of the parameters $S$ and $B$
in Algorithm \ref{Subsamplealgo}.

For images we observe a decrease in relative runtimes with higher resolution,
while the average relative error is independent of the image resolution.
In the point cloud case, however, the relative error increases slightly with
the instance size.
The number $S$ of sampled points seems to
considerably affect the relative error. An increase of the number of points
results in more accurate values, with average relative errors as low as about
$3 \%$ for $S = 4000$, while still maintaining a speedup of two orders of
magnitude on $128 \times 128$ images. Lower sample sizes yield 
higher average errors, but also lower runtimes. With $S = 500$ the
runtime is reduced by over four orders of magnitude with an average relative
error of less than $10 \%$. As to be expected, runtime increases linearly with
the number of repetitions $B$. However, the impact on the relative errors is
rather inconsistent. This is due to the fact, that the costs returned by
the subsampling algorithm are often overestimated, therefore averaging
over multiple tries does not yield improvements (see Figure \ref{fig:bias}).
This means that in order to increase the accuracy of the algorithm
it is advisable to keep $B = 1$ and instead increase the sample size $S$.
However, increasing $B$ can be useful to lower the variability of the results.
\begin{figure}
  \centering
  \includegraphics[width=0.49\textwidth]{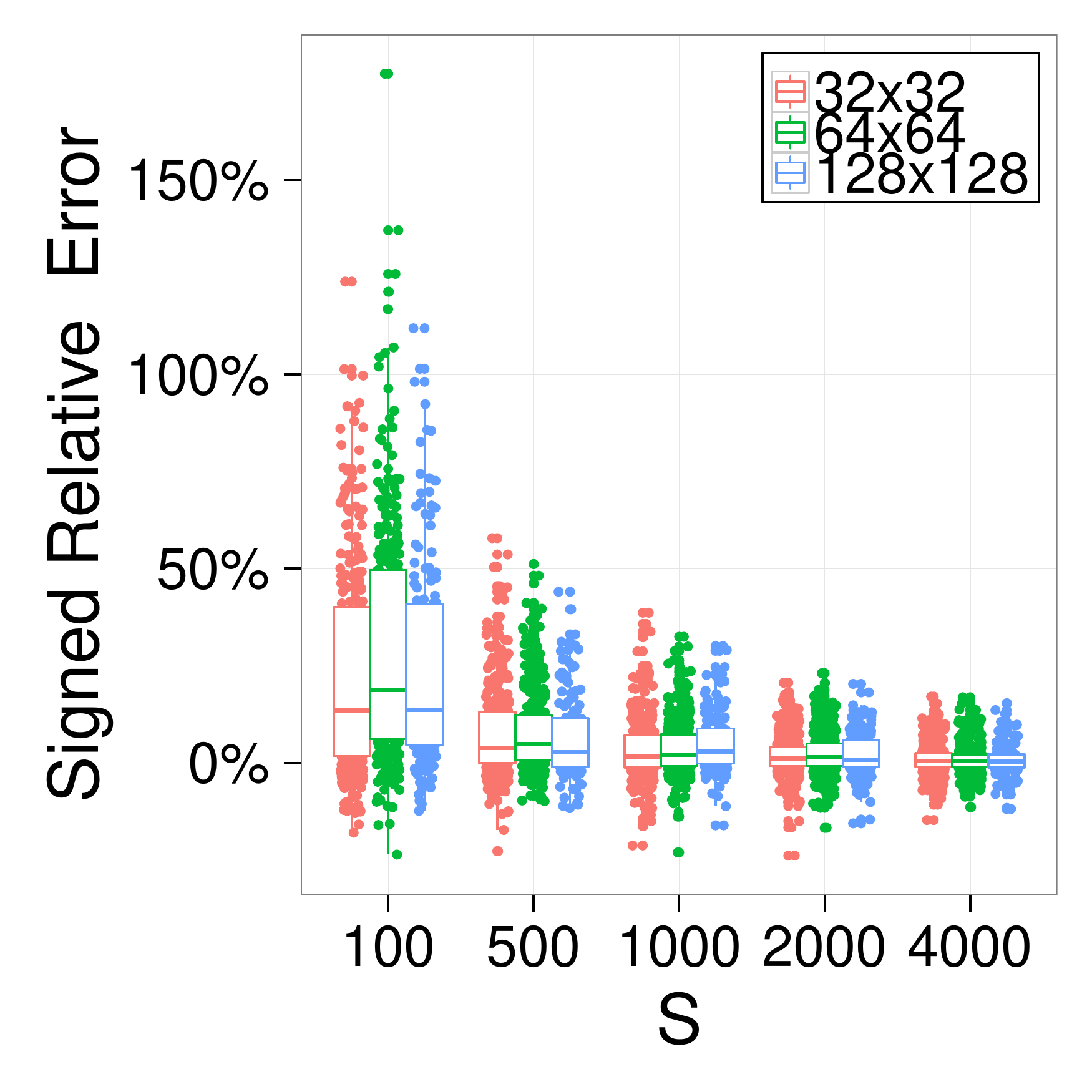}
  \caption{The signed relative approximation error $\left( \WhnmB_p(\br, \bs) -
    W_p(\br, \bs)\right) / W_p(\br, \bs)$ showing that the
  approximation overestimates the exact distance for small $S$ but the bias
  vanishes for
larger $S$. }
  \label{fig:bias}
\end{figure}

On the contrary, there is a big difference in accuracy between the image
classes. While Algorithm \ref{Subsamplealgo} has consistently low relative errors on the Cauchy Density images,
the exact optimal costs for White Noise images cannot be
approximated as reliably. The relative errors fluctuate more and are generally
much higher, as one can see from Figure
\ref{fig:classes} (left). In images with smooth structures and regular features the
subsamples are able to capture that structure and therefore deliver a more
precise representation of the images and a more precise value. This is not
possible in images that are very irregular or noisy, such as the White Noise
images, which have no structure to begin with. The Classic Images contain both
regular structures and more irregular regions,
therefore their relative errors are slightly higher than in the Cauchy Density cases.
The algorithm has a similar performance on the point cloud instances, that are modelled after the
Cauchy Density and Classic Images classes, while the Dirichlet instances have a more desirable
accuracy compared to the White Noise images, as seen in Figure \ref{fig:classes} (right).

\begin{figure*}
  \centering
  \includegraphics[width=0.49\textwidth]{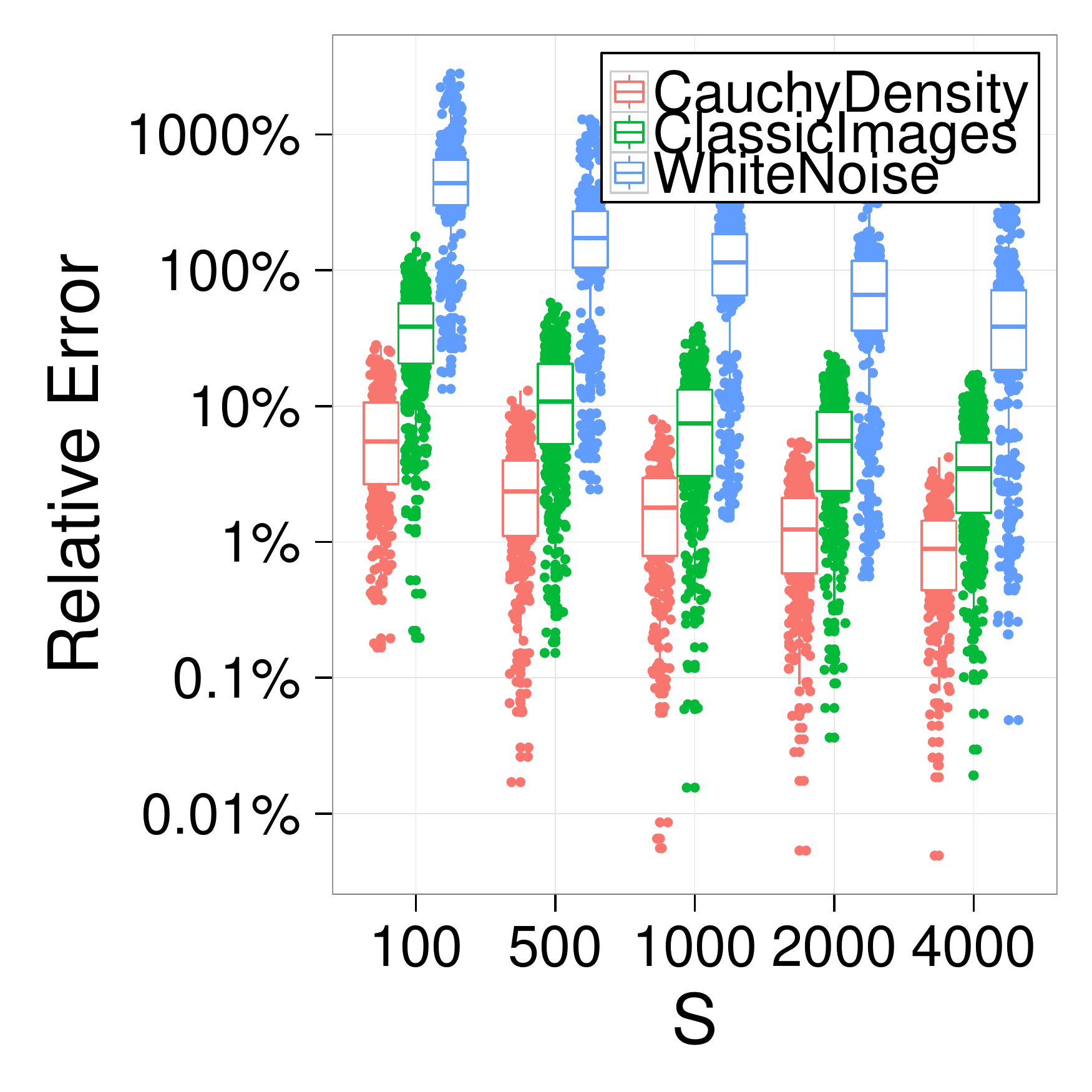}
  \includegraphics[width=0.49\textwidth]{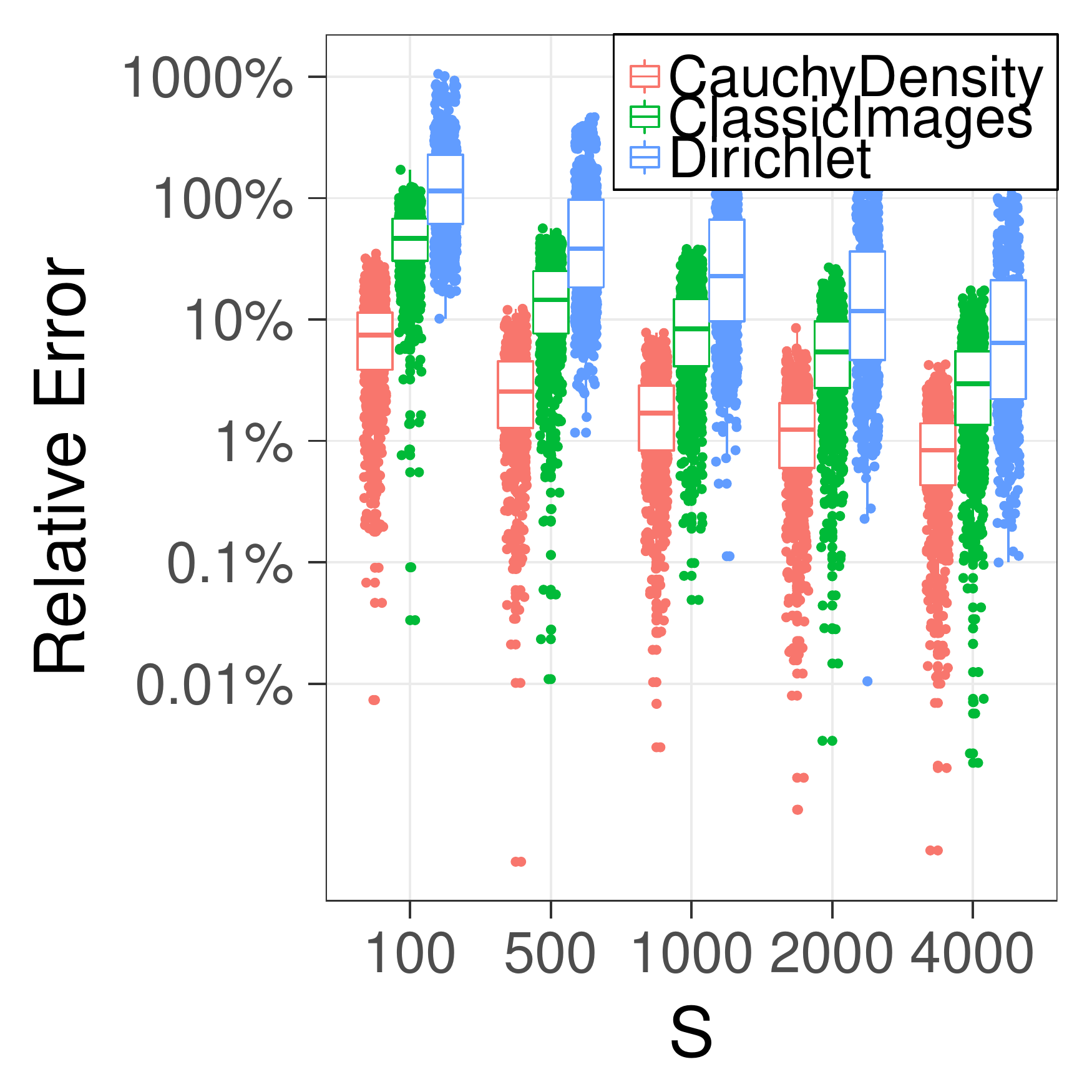}
  \caption{A comparison of the relative errors for different image classes (left) and and point cloud instance classes (right).}
  \label{fig:classes}
\end{figure*}

\begin{figure*}
  \centering
  \includegraphics[width=0.49\textwidth]{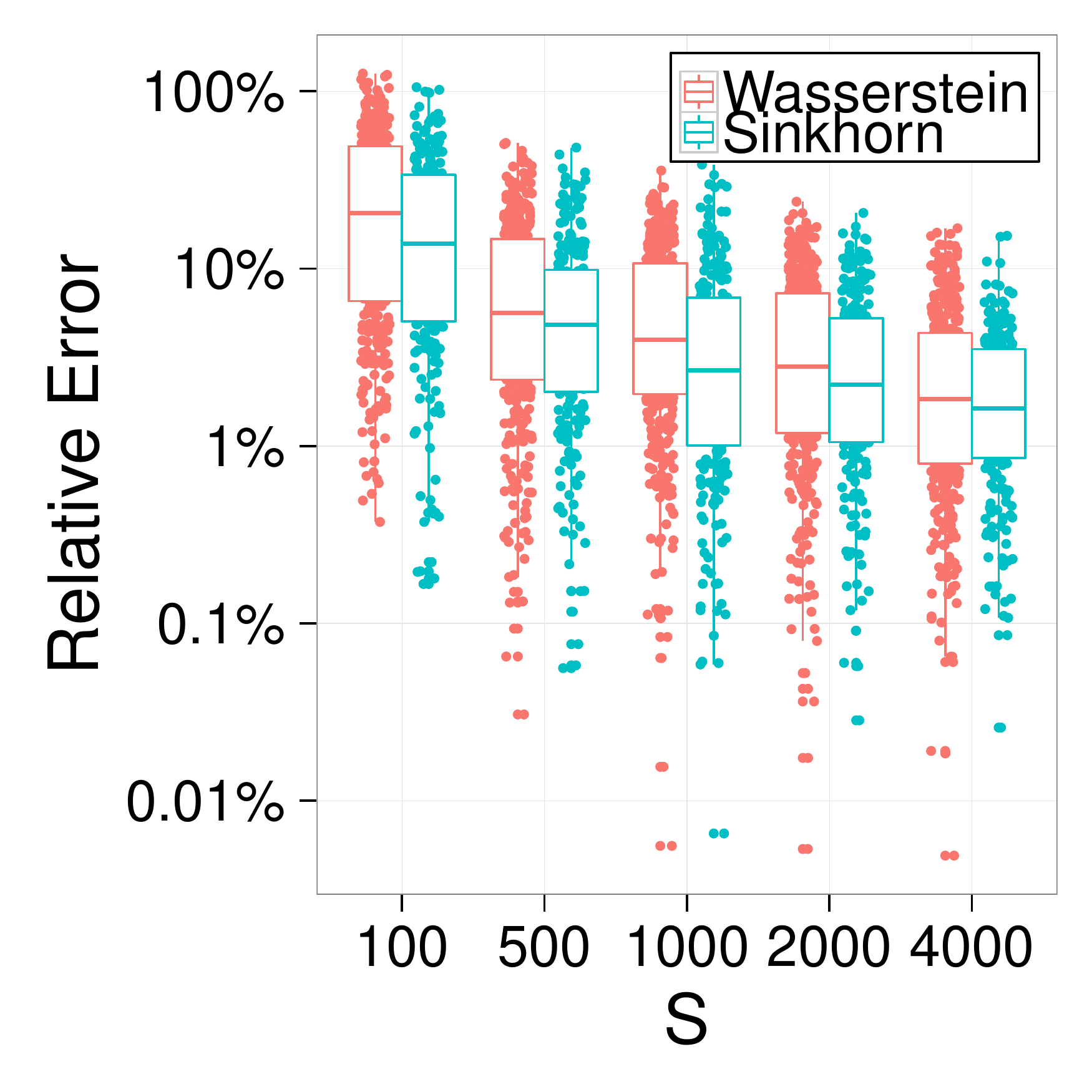}
  \caption{A comparison between the approximations of the Wasserstein and Sinkhorn distances.}
  \label{fig:solver}
\end{figure*}

There are no significant differences in performance between the different back-end solvers for the Wasserstein distance.
As Figure \ref{fig:solver} shows, accuracy seems to be better for
the Sinkhorn distance compared to the other three solvers which report the exact Wasserstein distance.

\begin{figure*}
  \centering
  \includegraphics[width=0.49\textwidth]{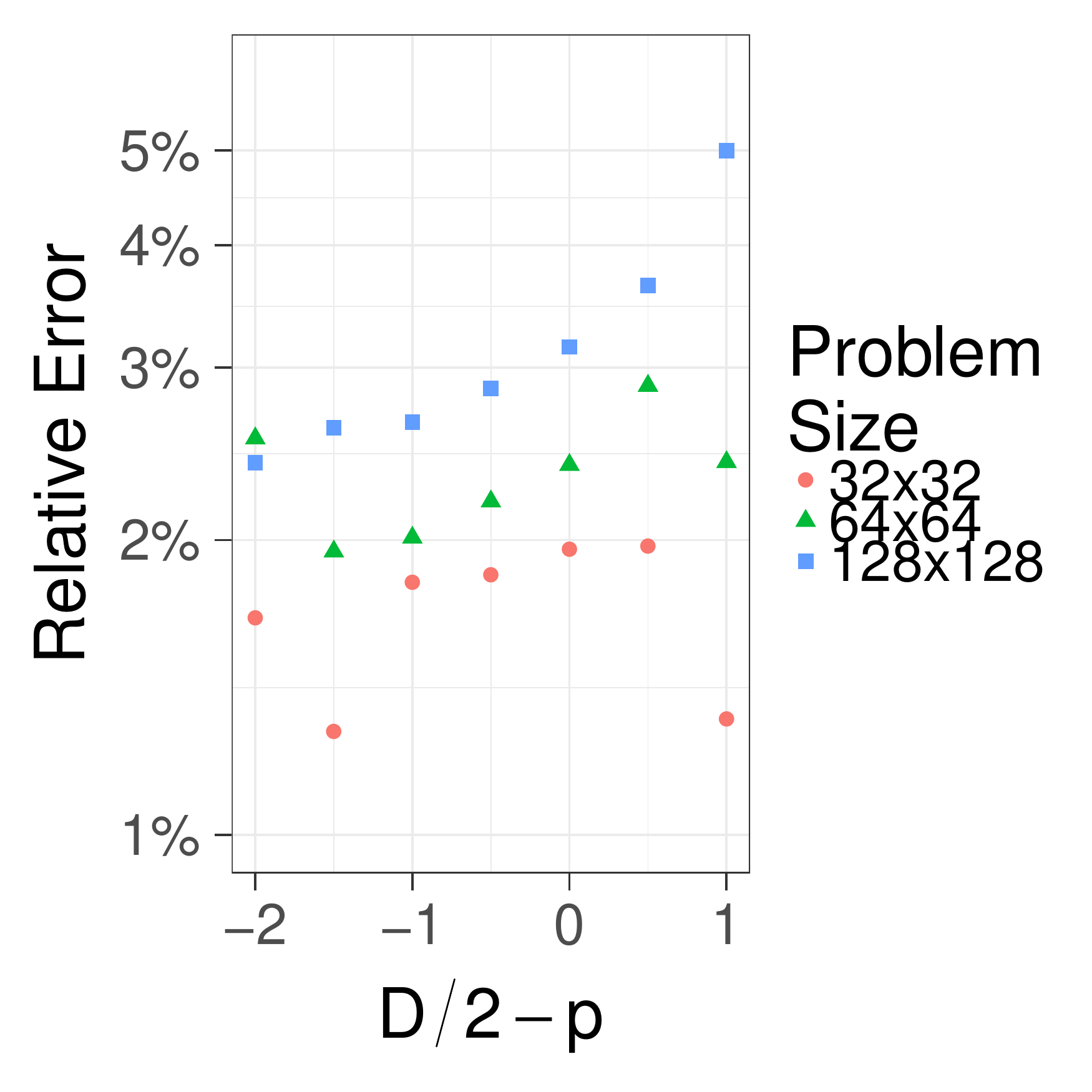}
  \caption{A comparison of the mean relative errors in the point cloud instances with sample size $S = 4000$ for different values of $p^\prime = (\nicefrac D 2) - p$.}
  \label{fig:pprime}
\end{figure*}
In the results of the point cloud instances we can observe the influence of the value $p^\prime = (\nicefrac D 2) - p$ on the scaling of the relative error with the instance size $N$ for constant sample size ($S = 4000$). This is shown in Figure \ref{fig:pprime}. We observe an increase of the relative error with $p^\prime$, as expected from the theory. However, we are not able to clearly distinguish between the three cases $p^\prime < 0$, $p^\prime = 0$ and $p^\prime > 0$. This might be due to the relatively small instance sizes $N$ in the experiments. While we see that the relative errors are independent of $N$ in the image case (compare Figure \ref{fig:overview}), for the point clouds $N$ has an influence on the accuracy that depends on $p^\prime$.

\section{Discussion}
As our simulations demonstrate, subsampling is a simple yet powerful tool to obtain good approximations to Wasserstein distances with only a small fraction of required runtime and memory.  It is especially remarkable that in the case of two dimensional images for a fixed amount of subsampled points, and therefore a fixed amount of time and memory, the relative error is independent of the resolution/size of the images. Based on these results, we expect the subsampling algorithm to return similarly precise results with even higher resolutions of the images it is applied to, while the effort to obtain them stays the same. Even in point cloud instances the relative error only scales mildly with the original input size $N$ and is dependent on the value $p^\prime$.

The numerical results (Figure \ref{fig:overview}) show an inverse polynomial decrease of the approximation error with $S$, in accordance with the theoretical results.  In fact, the rate $\mathcal O(S^{-1/2p})$ is optimal.  Indeed, when $\br=\bs$ (are nontrivial measures), \cite{sommerfeld_inference_2018} show that $Z_S=S^{1/2p}[W_p(\brh_S,\bsh_S)-W_p(\br,\bs)]$ has a nondegenerate limiting distribution $Z$.  For each $R>0$ the function $x\mapsto \min(R,|x|)$ is nonnegative, continuous and bounded, so
\[
\liminf_{S\to\infty} E\{S^{1/2p}|W_p(\brh_S,\bsh_S)-W_p(\br,\bs)|\}
=
\liminf_{S\to\infty} E\{|Z_S|\}\ge 
\liminf_{S\to\infty} E\min\{R,|Z_S|\}
=
E\min(R,|Z|).
\]
Letting $R\to\infty$ and using the monotone convergence theorem yields
\[
\liminf_{S\to\infty} E\{S^{1/2p}|W_p(\brh_S,\bsh_S)-W_p(\br,\bs)|\}
\ge E|Z|>0.
\]

When applying the algorithm, it is important to note that the quality of the returned values depends on the structure of the data. In very irregular instances it is necessary to increase the sample size in order to obtain similarly precise results, while in regular structures a small sample size suffices.

{Our scheme allows the parameters $S$ and $B$ to be easily tuned towards faster runtimes or more precise results, as desired. Increases and decreases of the sample size $S$ will increase/decrease the mean approximation of $W_p$ by $\hat W_p^{(S)}$, while $B$ will only affect the concentration around $E\hat W_p^{(S)}$.  Empirically, we found that for fixed computational cost, the best performance is achieved when $B=1$ (compare Figure~\ref{fig:overview}), suggesting that the bias is more dominant than the variance in the mean squared error.}

The scheme presented here can readily be applied to other optimal transport distances, as long as a solver is available, as we demonstrated with the Sinkhorn distance \citep{Cuturi2013a}. Empirically, we can report good performance in this case, suggesting that entropically regularized distances might be even more amenable to subsampling approximation than the Wasserstein distance itself. Extending the theoretical results to this case would require an analysis of the mean speed of convergence of empirical Sinkhorn distances, which is an interesting task for future research.

All in all, subsampling proves to be a general, powerful and versatile tool that can be used with virtually any optimal transport solver as back-end and has both theoretical approximation error guarantees, and a convincing performance in practice. It is a challenge to extend this method in a way which is specifically tailored to the geometry of the underlying space $\X$, which may result in further improvements.

\section*{Appendix}
\subsection{Proof of Theorem \ref{thm:mean_rate_null}}

\paragraph{Proof strategy}
The method used in this proof has been employed before to bound the mean rate of convergence of the empirical Wasserstein distance on a general metric space $(\X, d)$ \citep{Boissard2014,fournier_rate_2014}. In essence, it constructs a tree on the space $\X$ and bounds the Wasserstein distance with some transport metric in the tree, which can either be computed explicitly or bounded easily ({see also \cite{heinrich2018strong}, who use a coarse-graining tree in order to bound the Wasserstein distance in the context of mixture models}).  Our construction is specifically tailored to finite spaces, and allows to obtain a better dependence on $N=|\X|$ in Theorem~\ref{thm:grid} while preserving the rate $S^{-1/2}$.

More precisely, in our case of finite spaces, let $\T$ be a spanning tree on $\X$ (that is, a tree with vertex set $\X$ and edge lengths given by the metric $d$ on $\X$) and $d_\T$ the metric on $\X$ defined by the path lengths in the tree. Clearly, the tree metric $d_\T$  dominates the original metric $d$ on $\X$ and hence $W_p(\br, \bs) \leq W_p^\T(\br, \bs)$ for all $\br, \bs \in\mathcal{P}(\X)$, where $W_p^\T$ denotes the Wasserstein distance evaluated with respect to the tree metric. The goal is now to bound $E\left[(W_p^\T(\brh_S, \br))^p \right]$. {We refer to \cite{tameling2018computational} for examples and comparisons of different spanning trees on two-dimensional grids.} 

Assume $\T$ is rooted at $\root(\T) \in \X$. Then, for $x \in \X$ and $x \neq \root(\T)$
we may define $\parent(x) \in \X$ as the immediate neighbor of $x$ in the unique path
connecting $x$ and $\root(\T)$. We set $\parent(\root(\T)) = \root(\T)$.
We also define $\children(x)$ as the set of vertices $x^\prime \in \X$ such that
there exists a sequence $x^\prime = x_1, \dots, x_l = x \in \X$ with $\parent(x_j) = x_{j+1}$
for $j = 1, \dots, l -1$. Note that with this definition $x \in \children(x)$. Additionally,
define the linear operator $S_\T \colon \RR^\X \rightarrow \RR^\X$
\begin{equation}
	(S_\T \bu)_x = \sum_{x^\prime \in \children(x)} u_{x^\prime}.
	\label{eq:def_ST}
\end{equation}

\paragraph{Building the tree}
We build a $q$-ary tree on $\X$. To this end, we split $\X$ to $l_{\max}+2$ groups and build the tree in such a way that a node at level $l+1$ has a unique parent at level $l$ with edge length $q^{-l}$.  The formal construction follows.

For $l\in \{0, \dots,l_{\max}\}$ we let $Q_l\subset\X$ be the center points of a
$q^{-l}\diam(\X)$ covering of  $\X$, that is 
\[
  \bigcup _{x\in Q_l} B(x, q^{-l}\diam(\X)) = \X, \text{ and } |Q_l| =
  \mathcal{N}(\X, q^{-l}\diam(\X)),
\]
where $B(x, \ve) = \{x' \in \X : d(x,x') \leq \ve \}$.
Additionally set $Q_{l_{\max} + 1} = \X$.
Now define $\tilde{Q}_l = Q_l \times \{l\}$ and we will build a tree structure
on $\cup_{l = 0}^{l_{\max} + 1} \tilde{Q}_l$.

Since we must have $|\tilde{Q}_0| = 1$ we can take this element as the root.  
Assume now that the tree already
contains all elements of $\cup_{j=0}^l \tilde{Q}_j$.  Then,
we add to the tree all elements of $\tilde{Q}_{l+1}$  
by choosing for $(x, l+1)\in \tilde{Q}_{l+1}$ (exactly one) parent element $(x', l)\in \tilde{Q}_l$  such
that $d(x, x') \le q^{-l}\diam(\X)$. This is possible, since  $Q_l$ is a
$q^{-l}\diam(\X)$ covering of $\X$. We set the length of the connecting edge to
$q^{-l}\diam(\X)$.

In this fashion we obtain a spanning tree $\T$ of $\cup_{l=0}^{l_{\max} + 1}
\tilde{Q}_l$ and a partition
$\{\tilde{Q}_l\}_{l=0, \dots, l_{\max} + 1}$. 
About this tree we know that
\begin{itemize}
  \item it is in fact a tree. First, it is connected, because the construction
    starts with one connected component and in every subsequent step all
    additional vertices are connected to it. Second, it contains no cycles.
    To see this let $((x_1, l_1), \dots, (x_K,l_K))$ be a cycle in $\T$. Without loss of
    generality we may assume $l_1= \min\{l_1, \dots, l_K\}$.
    Then, $(x_1, l_1)$ must have at least two edges connecting it to vertices in a
    $\tilde{Q}_l$ with $l \geq l_1$ which is impossible by construction.
  \item $|\tilde{Q}_l|  = \mathcal{N}(\X, q^{-l}\diam(\X))$ for $0\leq l \leq
    l_{\max}$.
  \item $d(x, \parent(x)) = q^{-l + 1}\diam(\X)$ whenever $x\in \tilde{Q}_l$, $l\ge 1$. 
  \item  $d(x, x') \leq d_{\T} \left( (x,l_{\max} + 1), (x', l_{\max} + 1) \right)$.
\end{itemize}
Since the leaves of $\T$ can be identified with $\X$  a measure $\br \in
\mathcal{P}(\X)$ canonically defines  a probability measure $\br^\T\in
\mathcal{P}(\T)$ for which $r^\T_{(x, l_{\max} + 1)} = r_x$ and $r^\T_{(x,l)} =
0$ for $l\leq l_{\max}$.
In slight abuse of notation we will denote the measure $\br^\T$ simply by $\br$.
With this notation, we have 
$W_p(\br, \bs) \leq W_p^\T(\br, \bs)$
for all $\br, \bs \in\mathcal{P}(\X)$.
\paragraph{Wasserstein distance on trees}
Note also that $\T$ is \textit{ultra-metric} that is,
all its leaves are at the same distance from the root. For trees of this type, we
can define a height function $h:\X\ra[0,\infty)$ such that $h(x)=0$ if
$x\in\X$ is a leaf and $h(\parent(x)) - h(x) = d_\T(x,\parent(x))$ for
all $x\in\X\setminus \root(\X)$. There is an explicit formula
for the Wasserstein distance on ultra-metric trees 
\citep{kloeckner_geometric_2013}. Indeed, if $\br,\bs\in\mathcal{P}(\X)$ then
  \begin{equation}
       (W_p^{\T} (\br,\bs))^p 
       = 2^{p-1} \sum_{x\in\X} \left( h(\parent(x))^p -
      h(x)^p \right) \left| (S_\T\br)_x - (S_\T\bs)_x \right|,
    \label{eq:formula_ultra_metric}
  \end{equation}
  with the operator $S_\T$ as defined in \eqref{eq:def_ST}.
  For the tree $\T$ constructed above and $x\in \tilde{Q}_l$ with $l=0,\dots,
  l_{\max}$ we have 
  \[
    h(x) = \sum_{j=l}^{l_{\max}} q^{-j} \diam(\X),
  \]
  and therefore $\diam(\X)q^{-l} \leq h(x) \leq 2\diam(\X)q^{-l}$.  This yields
  \[
    (h(\parent(x))^p - (h(x))^p) \leq (\diam(\X))^p q^{-(l-2)p}. 
  \]
  Then  \eqref{eq:formula_ultra_metric}  yields
  \begin{align*}
     E\left[ W_p^p (\brh_S, \br) \right] 
     \leq 2^{p-1}q^{2p} (\diam(\X))^p
    \sum_{l=0}^{l_{\max} + 1} q^{-lp} \sum_{x\in \tilde{Q}_l}
    E|(S_\T\brh_S)_{x} - (S_\T\br)_{x}|.
  \end{align*}
  Since $(S_\T\brh_S)_{x}$ is the mean of $S$ i.i.d. Bernoulli variables
  with expectation $(S_\T\br)_{x}$ we have
  \begin{multline*}
    \sum_{x\in \tilde{Q}_l} E|(S_\T\brh_S)_{x} - (S_\T\br)_{x}| \leq
    \sum_{x\in \tilde{Q}_l} \sqrt{\frac{(S_\T\br)_{x}(1-(S_\T\br)_{x})}{S}} \\
     \leq \frac{1}{\sqrt{S}} \left( \sum_{x\in \tilde{Q}_l} (S_\T\br)_{x} \right)^{1/2}
    \left( \sum_{x\in \tilde{Q}_l}  (1-(S_\T\br)_{x})\right)^{1/2}    
    \leq \sqrt{|\tilde{Q}_l| / S},
  \end{multline*}
  using H\"older's inequality and the fact that $\sum_{x\in \tilde{Q}_l} (S_\T
  \br)_x = 1$ for all $l = 0, \dots, l_{\max} + 1$.  This finally yields
  \begin{multline*}
    E\left[ W_p^p (\brh_S, \br) \right] 
    \leq 
    2^{p-1}q^{2p} (\diam(\X))^p
    \left( q^{-(l_{\max} + 1)p} \sqrt{N} + \sum_{l=0}^{l_{\max}} q^{-lp}
    \sqrt{\mathcal{N}(\X, q^{-l}\diam(\X))} \right) / \sqrt{S} \\
    \leq \E_q(\X, p) / \sqrt{S}.
  \end{multline*}

\paragraph{Covering by arbitrary sets}
We now explain how to obtain the second formula for $\E_q$ as stated in Remark~\ref{rem:N1}. The idea is to define the coverings with arbitrary sets, not necessarily balls.  Let
\[
\mathcal N_1(\X,\delta)
=\inf\{m:\exists A_1,\dots,A_m\subseteq\X,\diam(A_i)\le 2\delta,\cup A_i\supseteq\X\}.
\]
Since balls satisfy the diameter condition, $\mathcal N_1\le \mathcal N$.  Furthermore, if $\X'\supseteq \X$, then $\mathcal N_1(\X,\delta)\le \mathcal N_1(\X',\delta)$, which is not the case for $\mathcal N$.  For example, let $\X=\{-1,1\}\subset \{-1,0,1\}=\X'$ and observe that
\[
\mathcal N_1(\X,1)=1=\mathcal N_1(\X',1),
\qquad \textrm{but}\qquad
\mathcal N(\X,1)=2>1=\mathcal N(\X',1).
\]

The tree construction with respect to the new covering numbers is done in a similar manner.  For each $0\le l\le l_{\max}$ let $Q_l'$ be a collection of disjoint sets of diameter $2q^{-l}\diam(\X)$ that cover $\X$ and $|Q_l'|=\mathcal N_1(\X,q^{-l}\diam(\X))$.  Let $Q_l=\{x_1,\dots,x_{|Q_l'|}\}\subseteq\X$ be an arbitrary collection of representatives from the sets in $Q_l'$.  Such representatives exist by minimality of $|Q_l'|$ and they are different by the disjoint nature of $Q_l'$.  Additionally set $Q_{l_{\max}+1}=\X$.  Construct the tree in the same way, except that now we only have the bound $d(x,x')\le 2q^{-l}\diam(\X)$ for $(x,l+1)\in \tilde Q_{l+1}$ and a corresponding $(x,l)\in \tilde Q_{l}$, so we need to set the edge length to be $2q^{-l}\diam(\X)$, twice as much as in the original construction.  The proof then goes in the same way, with an extra factor $2^p$.  We obtain an alternative bound
\[
    \E_q  = 
    2^{2p-1}q^{2p} (\diam(\X))^p
    \left( q^{-(l_{\max} + 1)p} \sqrt{N} + \sum_{l=0}^{l_{\max}} q^{-lp}
    \sqrt{\mathcal{N}_1(\X, q^{-l}\diam(\X))} \right).
\]
In comparison with \eqref{eq:Eq}, we replaced $\mathcal N$ by $\mathcal N_1$.  The price to pay for this is an additional factor of $2^p$.

  \subsection{Proof of Theorem \ref{thm:grid}}
  We may assume without loss of generality that $\X \subseteq [0,\diam(\X)]^D$.  The covering numbers of the cube with Euclidean balls behave badly in high dimensions, so it will prove useful to replace the Euclidean norm by the infinity norm $\|x\|_\infty=\max_i|x_i|$, $x=(x_1,\dots,x_D)\in\RR^D$.  With this norm we have $\mathcal N([0,\diam(\X)]^D,\epsilon\diam(\X),\|\cdot\|_\infty)\le (\lceil 1/(2\epsilon)\rceil)^D$.  If $q$ is an integer, then
  \[
    \mathcal{N}(\X, q^{-l}\diam(\X),\|\cdot\|_\infty)
    \le \mathcal{N}([0,\diam(\X)]^D,q^{-l}\diam(\X)/2,\|\cdot\|_\infty)
    \le \lceil q^l\rceil ^D
    =q^{lD}.
  \]
  This yields
\[
    \sum_{l=0}^{l_{\max}} q^{-lp}
    \sqrt{\mathcal{N}(\X, q^{-l}\diam(\X))} 
    \le
\sum_{l=0}^{l_{\max}}
    q^{l(D/2 - p)}
    =
    \begin{cases}
      (1 - q^{(l_{\max}+1)(D/2 - p)}) / (1 - q^{D/2-p}) & \text{if } D \neq 2p, \\
      l_{\max}+1 & \text{if } D = 2p.
    \end{cases}
\]
  Denote for brevity $p'=D/2-p$ and plug this into \eqref{eq:Eq}:
\[
S^{1/2}E\left[ W_p^p (\brh_S, \br,\|\cdot\|_\infty) \right]
\le 
    2^{p-1}q^{2p}(\diam(\X))^p
    \left[
    q^{-p(l_{\max}+1)}\sqrt N
    +
    \begin{cases}
      (1 - q^{(l_{\max}+1)p'}) / (1 - q^{p'}) & \text{if } p' \ne 0, \\
      l_{\max}+1 & \text{if } p'= 0.
    \end{cases}
    \right]
\]
  If $p'<0$, then let $l_{\max}\to\infty$.  Otherwise, choose $l_{\max}=\lfloor D^{-1}\log_qN\rfloor$ (giving the best dependence on $N$), so that the element inside the square brackets is smaller than
\begin{equation}\label{eq:CDpqN}
\begin{cases}
      1 / (1 - q^{p'}) & \text{if } p' < 0, \\
      2 + D^{-1}\log_qN & \text{if } p'= 0, \\
      N^{1/2 - p/D} + (N^{1/2 - p/D}q^{p'} - 1) / (q^{p'} - 1) & \text{if } p' > 0
    \end{cases}
    \quad \le \quad
    \begin{cases}
      1 / (1 - q^{p'}) & \text{if } p' < 0, \\
      2 + D^{-1}\log_qN & \text{if } p'= 0, \\
      (2q^{p'}-1)N^{1/2 - p/D}/(q^{p'} - 1) & \text{if } p' > 0.
    \end{cases}
\end{equation}
The right-hand side is $C_{D,p}(N)$ for $q=2$.  To get back to the Euclidean norm use $\|a\|_2\le \|a\|_\infty \sqrt D$, so that
\[
     E\left[ W_p^p (\brh_S, \br) \right] 
     \le D^{p/2} E\left[ W_p^p (\brh_S, \br,\|\cdot\|_\infty) \right] 
     \le D^{p/2} 2^{p-1}q^{2p}(\diam(\X))^p C_{D,p}(N)/\sqrt S,
\]
which is the desired conclusion.

\begin{lem}\label{lem:q2}
\begin{enumerate}

\item[(a)] Let $\tilde C_{D,p}(q,N)$ denote the right-hand side of \eqref{eq:CDpqN}.  Then the minimum of the function $q\mapsto q^{2p}\tilde C_{D,p}(q,N)$ on $[2,\infty)$ is attained at $q=2$.

\item[(b)] Let $q\ge2$, $p,D$ integers, and $p'=D/2-p$.  If $p'<0$, then $1 / (1 - q^{p'} )\le 2+\sqrt2$ and if $p'>0$, then $2+1/(q^{p'} - 1)\le 3+\sqrt2$.
\end{enumerate}
\end{lem}
\begin{proof}
We begin with (b).  If $p'<0$ then $1/(1-q^{p'})$ is decreasing in $q$ and increasing in $p'$.  The integer constaints on $D$ and $p$ imply that the maximal value $p'$ can attain is $-0.5$.  The smaller value $q$ can attain is 2.  Thus
\[
1/(1-q^{p'})
\le 1/(1-2^{-0.5})
=\frac{\sqrt 2}{\sqrt 2-1}
=\sqrt2(\sqrt2+1)
=2+\sqrt 2.
\]
When $p'>0$ the term $2+1/(q^{p'} - 1)$ is decreasing in $p'\ge0.5$ and in $q\ge2$, so it is bounded by
\[
2+1/(\sqrt2-1)
=3+\sqrt2.
\]

To prove (a) we shall differentiate the function $q^{2p}\tilde C_{D,p}(q,N)$ with respect to $q$ and show that the derivative is positive for all $q\ge2$, and $p,D,N\ge1$.

For negative $p'$ consider the function
\[
f_1(q)= \frac{q^{2p}} {1-q^{p'}},
\qquad q\ge 2;p\ge 1;p'<0.
\]
Its derivative is
\[
f_1'(q)
=\frac{2pq^{2p-1}(1-q^{p'})+p'q^{p'-1}q^{2p}}{(1-q^{p'})^2}
=\frac{q^{2p-1}}{1-q^{p'}}
\left[2p + \frac{p'q^{p'}}{1-q^{p'}}\right]
.
\]
It suffices to show that the term in square brackets is positive, since $1-q^{p'}>0$. Let us bound $q^{p'}$ and the denominator $(1-q^{p'})^{-1}$.  Since $e^x\ge 1+x$ for $x\ge0$, $e^{-x}\le 1/(1+x)$ and setting $x=-p'\log q$ gives
\[
q^{p'}
=e^{p'\log q}
\le \frac1{1-p'\log q}.
\]
Hence
\[
1-q^{p'}
\ge 1-\frac1{1-p'\log q}
= \frac{1-p'\log q-1}{1-p'\log q}
= \frac{-p'\log q}{1-p'\log q}.
\]
so that
\[
\frac{q^{p'}}{1-q^{p'}}
\le \frac1{1-p'\log q}\frac{1-p'\log q}{-p'\log q}
= \frac{1}{-p'\log q}.
\]
Conclude that, since $p'<0$,
\[
2p + \frac{p'q^{p'}}{1-q^{p'}}
\ge 2p + p'\frac{1}{-p'\log q}
=2p + \frac1{-\log q}
=2p - \frac1{\log q}
\ge 2p - \frac1{\log 2}
\ge 2 - \frac1{\log 2}
>0.
\]
For $p'=0$ consider the function
\[
f_2(q)
=q^{2p}(2+D^{-1}\log_qN)
=2q^{2p}+\frac{q^{2p}\log N}{D\log q},
\qquad q\ge2; D=2p\ge2.
\]
Its derivative is
\[
f_2'(q)=4pq^{2p-1}+\frac{\log N}{D(\log q)^2}\left[2pq^{2p-1}\log q-q^{-1}q^{2p}\right]
=q^{2p-1}\left[4p+\frac{\log N}{D(\log q)^2}(2p\log q-1)\right]
>0
\]
since $2p\log q\ge 2\log 2>1$.

For $p'>0$ consider the function
\[
f_3(q)
=q^{2p}[2+1/(q^{p'}-1)]
=2q^{2p}+\frac{q^{2p}}{q^{p'}-1}
=2q^{2p} - f_1(q),
\qquad q\ge2; p\ge1; p'>0.
\]
The derivative is
\[
4pq^{2p-1}-\frac{q^{2p-1}}{1-q^{p'}}
\left[2p + \frac{p'q^{p'}}{1-q^{p'}}\right]
=4pq^{2p-1}+\frac{q^{2p-1}}{q^{p'}-1}
\left[2p - \frac{p'q^{p'}}{q^{p'}-1}\right]
.
\]
This function is more complicated and we need to split into cases according to small, large or moderate values of $p'$.

\textbf{Case 1:  $p'\le 0.5$.}  Then the negative term can be bounded using $q^{p'}-1\ge p'\log q$ as
\[
\frac{p'q^{p'}}{q^{p'}-1}
=p'+\frac{p'}{q^{p'}-1}
\le p'+\frac1{\log q}
\le p'+\frac1{\log 2}
\le 0.5 + \frac1{\log 2}
<2\le 2p.
\]
Thus $f_3'(q)\ge0$ in this case.

To deal with larger values of $p'$ rewrite the derivative as
\[
q^{2p-1}\left[4p+\frac{2p}{q^{p'}-1} - \frac{p'q^{p'}}{(q^{p'}-1)^2}\right],
\]
and bound the negative part:
\[
\frac{p'q^{p'}}{(q^{p'}-1)^2}
=\frac{p'}{q^{p'}-1} + \frac{p'}{(q^{p'}-1)^2}
\le \frac1{\log q} + \frac1{(q^{p'}-1)\log q}.
\]
\textbf{Case 2: $p'\ge1$.}  Then $q^{p'}-1\ge1$ so this is smaller than
\[
\frac 1{\log 2} +\frac1{\log 2}
=\frac 2{\log 2}
<4\le 4p.
\]
Hence the derivative is positive in this case.

\textbf{Case 3: $p'\ge1/2$ and $q\ge e$.}  Then this is smaller than
\[
1+\frac1{e^{1/2} - 1}
\le 1+\frac1{\sqrt 2 -1}
=2+\sqrt 2 < 4\le 4p.
\]
Hence the derivative is positive in this case.

\textbf{Case 4: $q\le e$ and $p'\in[1/2,1]$.}  The negative term is bounded by
\[
\frac1{\log q} + \frac1{(q^{p'}-1)\log q}
\le \frac1{\log 2} + \frac1{(q^{p'}-1)\log 2}
\le \frac1{\log 2} + \frac1{(\sqrt2-1)\log 2}
=\frac{2+\sqrt 2}{\log 2}
\approx 4.93,
\]
whereas the positive term can be bounded below as
\[
4p+\frac{2p}{q^{p'}-1}
\ge 4 + \frac2{e-1}
\approx 5.16
>4.93.
\]
This completes the proof.

\end{proof}

  \subsection{Proof of Theorem \ref{thm:simple_concentration}}
  We introduce some additional notation. For
  $(x,y), (x',y')\in\X^2$ we set
  \[
    d_{\X^2}((x,y), (x',y')) = \left\{d^p(x,x') + d^p(y,y')\right\}^{1/p}
  \]
  We further define the function $Z:(\X^2)^{SB} \ra \RR$ via
  \[
    \begin{split}
    \left( (x_{11}, y_{11}), \dots , (x_{SB}, y_{SB}) \right)
     \mapsto
    \frac{1}{B}\sum_{i=1}^B \left[ W_p\left(\frac{1}{S}\sum_{j=1}^S
      \delta_{x_{ji}}, \frac{1}{S}\sum_{j=1}^S
      \delta_{y_{ji}}\right) - W_p(\br, \bs)
    \right].
  \end{split}
  \]
  Since $W_p^p(\cdot, \cdot)$ is jointly convex 
  {\citep[Theorem~4.8]{villani_optimal_2008}},
  \begin{equation*}
    W_p\left(\frac{1}{S}\sum_{j=1}^S \delta_{x_j}, \frac{1}{S}\sum_{j=1}^S
    \delta_{y_j}\right) \leq \left\{ \frac{1}{S} \sum_{j=1}^S W_p^p(\delta_{x_j},
    \delta_{y_j}) \right\}^{1/p}
    = 
    S^{-1/p} \left\{\sum_{j=1}^S d^p(x_j, y_j)\right\}^{1/p}.
  \end{equation*}

  Our first goal is to show that $Z$ is Lipschitz continuous. To this end, let
  $((x_{11}, y_{11}) ,\dots, (x_{SB}, y_{SB}))$ and $((x'_{11}, y'_{11}) ,\dots,
  (x'_{SB}, y'_{SB}))$ arbitrary elements of
  $(\X^2)^{SB}$. Then, using the reverse triangle
  inequality and the relations above
  \begin{align*}
    &|Z((x_{11}, y_{11}) ,\dots, (x_{SB}, y_{SB})) 
     - Z((x'_{11}, y'_{11}) ,\dots, (x'_{SB}, y'_{SB}))| 
    \\ &\leq \frac{1}{B}\sum_{i=1}^B \bigg|W_p\left(\frac{1}{S}\sum_{j=1}^S
    \delta_{x_{ji}}, \frac{1}{S}\sum_{j=1}^S
    \delta_{y_{ji}}\right)  
    -  W_p\left(\frac{1}{S}\sum_{j=1}^S \delta_{x'_{ji}}, \frac{1}{S}\sum_{j=1}^S
    \delta_{y'_{ji}}\right)\bigg| 
    \\ & \leq \frac{1}{B}\sum_{i=1}^B \bigg[W_p\left(\frac{1}{S}\sum_{j=1}^S
      \delta_{x_{ji}}, \frac{1}{S}\sum_{j=1}^S
      \delta_{x'_{ji}}\right)  
       +  W_p\left(\frac{1}{S}\sum_{j=1}^S \delta_{y_{ji}}, \frac{1}{S}\sum_{j=1}^S
    \delta_{y'_{ji}}\right)\bigg] \\
    &\leq \frac{S^{-1/p}}{B} \sum_{i=1}^B \bigg[ \left\{ \sum_{j=1}^S d^p(x_{ji},
      x'_{ji})\right\}^{1/p} 
      +  \left\{ \sum_{j=1}^S d^p(y_{ji},
    y'_{ji})\right\}^{1/p}\bigg]
    \\ &\leq \frac{S^{-1/p}}{B} \left( 2B \right)^{\frac{p-1}{p}} \left\{ 
      \sum_{i,j}d^p_{\X^2}((x_{ji}, y_{ji}), (x'_{ji}, y'_{ji}))
    \right\}^{1/p}  
  \end{align*}
  Hence, $Z/2$ is Lipschitz continuous with constant $(SB)^{-1/p}$ relative to
  the $p$-metric generated by $d_{\X^2}$ on $(\X^2)^{SB}$.

  For $\tilde{\br}\in\mathcal{P}({\X^2})$ let $H(\cdot\:|\:\tilde{\br})$ denote
  the relative entropy with respect to $\tilde{\br}$. Since $\X^2$ has $d_{\X^2}$-diameter $2^{1/p}\diam(\X)$, we have by
  {\citet[Particular case $2.5$, page $337$]{bolley_weighted_2005}} that for every $\tilde{\bs}$
  \begin{equation}
    W_p(\tilde{\br}, \tilde{\bs}) \leq \left( 8\diam(\X)^{2p} H(\tilde{\br}\:|\:
    \tilde{\bs} )\right)^{1/2p}.
    \label{eq:entropy}
  \end{equation}

  If $X_{11},\dots,X_{SB}\sim\br$ and $Y_{11},\dots,Y_{SB}\sim\bs$ are all
  independent, we have 
  \[
    Z((X_{11}, Y_{11}), \dots , (X_{SB}, Y_{SB})) \sim \WhnmB_p(\br, \bs) -
    W_p(\br, \bs).
  \]
  The Lipschitz continuity of $Z$ and the transportation inequality
  \eqref{eq:entropy} yields a concentration result for this random variable. In
  fact, by {\citet[Lemma 6]{gozlan_large_2007}} we have
  \[
      P\bigg[ \WhnmB_p(\br, \bs) -
        W_p(\br, \bs) 
         \geq E\left[ \WhnmB_p(\br, \bs) -
        W_p(\br, \bs) \right] + z \bigg] 
        \leq \exp\left(
        \frac{-SBz^{2p}}{8\:\diam(\X)^{2p}} \right).
    \]
    for all $z\geq 0$. Note that $-Z$ is Lipschitz continuous as well and hence,
    by the union bound,
    \[
        P\bigg[ \left|\WhnmB_p(\br, \bs) -
          W_p(\br, \bs)\right| 
          \geq E\left[ \left|\WhnmB_p(\br, \bs) -
          W_p(\br, \bs) \right|\bigg] + z \right] 
          \leq 2\exp\left(
          \frac{-SBz^{2p}}{8\:\diam(\X)^{2p}} \right).
      \]
      Now, with  the
  reverse triangle inequality, Jensen's inequality and  Theorem~\ref{thm:mean_rate_null},
  \begin{multline*}
    E\left[ \left| \WhnmB_p(\br, \bs) - W_p(\br, \bs)\right| \right] \leq 
    E\left[ W_p(\brh_S, \br) + W_p(\bsh_S, \bs) \right]
    \\ \leq E\left[ W_p^p(\brh_S, \br) \right]^{1/p} + \left[W_p^p(\bsh_S, \bs)
    \right]^{1/p}
    \leq 2\E_q^{1/p} / S^{1/(2p)}.
  \end{multline*}
  Together with the last concentration inequality above, this concludes the
  proof of Theorem~\ref{thm:simple_concentration}.

\subsection{Proof of Theorem~\ref{thm:mean_squared_alt}}
Denote $V=| \WhnmB_p(\br, \bs) - W_p(\br, \bs)|$, $C=2\E_q^{1/p}/S^{1/(2p)}\ge0$, and observe that
\begin{align*}
    E\left[ V ^2\right] 
&    =\int_0^\infty P(V>\sqrt t)dt
    =2\int_0^\infty P(V>s)sds
    =2\int_{-C}^\infty P(V>z+C)(z+C)dz\\
    & \le 2\int_{-C}^C(z+C)dz + 4\int_C^\infty P(V>z+C)zdz
    \le 4C^2 + 8\int_C^\infty z\exp\left(
        -\frac{SBz^{2p}}{8\:\diam(\X)^{2p}} \right)dz
\end{align*}
by Theorem~\ref{thm:simple_concentration}.  Changing variables and using the inequality $y^{2p}\ge y^2$ (valid for $y,p\ge1)$ gives
\begin{align*}
8\int_C^\infty z\exp\left(
        -\frac{SBz^{2p}}{8\:\diam(\X)^{2p}} \right)dz
&= 8C^2\int_1^\infty y\exp\left(
        -\frac{SB(Cy)^{2p}}{8\:\diam(\X)^{2p}} \right)dy\\
\le 8C^2\int_1^\infty y\exp\left(
        -\frac{SBC^{2p}y^2}{8\:\diam(\X)^{2p}} \right)dy
        &=8C^2\frac{4(\diam(\X))^{2p}}{SBC^{2p}}
        \exp\left(
        -\frac{SBC^{2p}}{8\:\diam(\X)^{2p}}\right)\\
        &=4C^2\frac{(\diam(\X))^{2p}}{2^{2p-3}\E_q^2B}
        \exp\left(
        -\frac{4^p\E_q^2B}{8\:\diam(\X)^{2p}}\right),
\end{align*}
where we have used $C^2=4\E_q^{2/p}S^{-1/p}$.  Deduce that
\[
    E\left[ \left| \WhnmB_p(\br, \bs) - W_p(\br, \bs) \right| ^2\right] 
    \le 16\E_q^{2/p}\left\{1 + \frac{(\diam(\X))^{2p}}{2^{2p-3}\E_q^2B}
        \exp\left(
        -\frac{4^p\E_q^2B}{8\:\diam(\X)^{2p}}\right)\right\}S^{-1/p}
        \le 18\E_q^{2/p} S^{-1/p}.
\]
For the last inequality, note that \eqref{eq:Eq} implies $\E_q^2\ge 2^{6p-2}[\diam(\X)]^{2p}$ and hence $[\diam(\X)]^{2p}/[B2^{2p-3}\E_q^2]\le 2^{5-8p}\le 1/8$, so the term in parentheses is smaller than $1+1/8$.

Similar computations show that $E\left[ \left| \WhnmB_p(\br, \bs) - W_p(\br, \bs) \right| ^\alpha\right]=\mathcal O(S^{-\alpha/(2p)})$  for all $0\le \alpha\le 2p$.

\bibliographystyle{plainnat}
\bibliography{finitewasser}

\end{document}